\documentclass[aps,prx,preprint,groupedaddress]{revtex4-2}

\usepackage{graphicx}
\graphicspath{{./pics/}}
\usepackage{subcaption}
\usepackage{dcolumn}
\usepackage{bm}

\usepackage{amsmath,amssymb}
\usepackage[utf8]{inputenc}
\usepackage[T1]{fontenc}
\usepackage{mathptmx}
\usepackage[cal=cm]{mathalpha}
\usepackage{etoolbox}
\usepackage{mathtools}
\usepackage{amsthm}
\usepackage{xcolor}
\usepackage[normalem]{ulem}
\newtheorem{theorem}{Theorem}

\theoremstyle{definition}
\newtheorem{example}{Example}

\newcommand{\bx}{\mathbf{x}}
\newcommand{\bv}{\mathbf{v}}
\newcommand{\R}{\mathbb{R}}

\begin{document}


\title{From Flow-Box Invariants to Global Conservation Laws: \\
Machine Learning--Assisted Detection of Singularities and Near-Integrability}


\author{Jimmie Adriazola}
\email[]{jimmie.adriazola@asu.edu}
\affiliation{School of Mathematical and Statistical Sciences, Arizona State University, Tempe, AZ
85287, USA}

\author{Ioannis G.~Kevrekidis}
\email[]{yannisk@jhu.edu}
\affiliation{Department of Applied Mathematics and Statistics, Johns Hopkins University, Baltimore, MD 21218, USA}

\author{Panayotis G.~Kevrekidis}
\email[]{kevrekid@umass.edu}
\affiliation{Department of Mathematics and Statistics, University of Massachusetts Amherst, Amherst, MA 01003-4515, USA}

\author{Wei Zhu}
\email[]{Contact author: weizhu@gatech.edu}
\affiliation{School of Mathematics, Georgia Institute of Technology, Atlanta, GA 30332, USA}


\date{\today}

\begin{abstract}
Conservation laws reveal invariant structure in dynamical systems and
constrain their evolution to lower-dimensional level sets. Locally,
however, per the flow-box theorem, every $d$-dimensional system generically admits $d-1$ functionally
independent conservation laws near a regular point even though,
generically, only a few such
laws persist globally. 
Data-driven methods typically learn
conservation laws on a prescribed region of phase space from which
training samples are drawn, leaving their behavior beyond that region
unresolved. We formulate this local-to-global continuation problem and
develop an {expanding-domain neural deflation} framework that tracks the
number of independent conservation laws over nested domains. The
method recovers the expected $d-1$ local invariants and detects their
successive breakdown as the domain expands, relaxing to the expected
global conservation law limit. It distinguishes
the local invariants (associated with Frobenius integrability)
from persistent global structure and shows that
the domain scale over which additional laws survive reflects proximity
to  Liouville integrability. Pointwise losses can localize isolated obstructions,
whereas in more complicated dynamics the breakdown of our loss functions is gradual and
spatially diffuse. The framework thus provides a data-driven approach
for studying the persistence and breakdown of (so-far-learned) conservation laws across
phase space.
\end{abstract}


\maketitle

\section{\label{sec:intro}Introduction}

Conservation laws play a central role in the analysis of dynamical
systems~\cite{goldstein2002classical}. By restricting trajectories to invariant level sets, they
reduce the effective dimension of the dynamics, reveal hidden
geometric structure, and provide important information about
long-time behavior. In Hamiltonian mechanics, the existence of
sufficiently many independent first integrals in involution underlies
Liouville integrability \cite{arnold1989classical}. For systems arising
from variational principles, continuous symmetries give rise to
conserved quantities through Noether's theorem
\cite{noether1983invariante,olver1993applications}.

When the governing equations and their structure are known,
conservation laws may sometimes be derived analytically. For complex
systems, however, such derivations can be challenging, and  the
number of independent conservation laws may not be known in advance.
This has motivated the recent development of {\it data-driven approaches} for
discovering conserved quantities from trajectories, vector-field
evaluations, symmetries, or the geometry of invariant level sets
\cite{wetzel2020discovering,liu2021machine,ha2021discovering,
liu2022machine,lu2023discovering}. Methods aimed specifically at recovering a complete set of functionally independent conservation laws (for Hamiltonian systems) include,
among others, neural deflation~\cite{zhu2023machine} and its trajectory-based extension~\cite{chen2025data}.

Existing data-driven methods typically formulate conservation-law discovery on a prescribed domain of phase space~\cite{liu2021machine,liu2022machine,zhu2023machine,chen2025data}. This restriction arises naturally because training requires samples from the phase space, which are typically drawn from a distribution with compact support. Consequently, the learned conservation laws are identified only on the region 
representatively 
covered by the training distribution. Unless this region coincides with the entire
phase space ---which is particularly challenging computationally, especially
as the number of degrees of freedom increases--- the recovered quantities should therefore be regarded as
local conservation laws and need not extend globally. Moreover,
learning conservation laws on a single prescribed domain does not
reveal how they persist or break down as the domain is enlarged. In
particular, it does not identify the scales at which individual
conservation laws cease to exist, nor the geometric or dynamical
obstructions responsible for their loss.

This dependence on the domain is essential because conservation laws
have fundamentally different local and global characters. Indeed, near any
regular point of a $d$-dimensional dynamical system, the flow-box
theorem (also called the rectification or straightening theorem) guarantees the existence of $d-1$ functionally independent
local conservation laws \cite{lee2003smooth}. This conclusion requires
neither a Hamiltonian structure nor global integrability. Thus, even a
system with no nontrivial global conservation law possesses the
maximal number of independent conservation laws in a sufficiently
small neighborhood of every regular point.

Local existence is therefore universal and, by itself, does not
distinguish an integrable system from a non-integrable one. The more
informative question is how the locally available conservation laws
behave {\it as the region of interest expands}. Some may persist over
increasingly large domains and ultimately extend to global
conservation laws. Others may cease to admit a smooth or single-valued
continuation, lose functional independence, or break down when the
domain encounters an equilibrium, a singular set, a recurrence or a more subtle
global obstruction. From this perspective, as the box of phase space under consideration is
expanded,
the transition from local
to global conservation laws is reflected in the successive loss of
the invariants  initially (meaning: in the neighborhood of a regular point) guaranteed by the flow-box theorem.

This observation motivates the central problem considered in this
paper:
\begin{quote}
\itshape
Starting from a regular point in phase space, can one
recover a complete family of independent local conservation laws and
track how this family changes as the domain expands?
\end{quote}
More specifically,
we seek to determine when the number of independent conservation laws
decreases, to identify which local structures persist over large
domains, and, whenever possible, to localize and hopefully even rationalize 
the regions responsible
for their breakdown.

We investigate this problem by extending the neural deflation method~\cite{zhu2023machine,chen2025data}, originally developed to discover global conservation laws in Hamiltonian systems, to a nested family of expanding domains. Deflation, which traces its origins to root-finding for algebraic equations with multiple solutions~\cite{BrownGearhart1971,ojika1987modified,burrage1998deflation} and is now widely used to compute multiple steady states of nonlinear PDEs~\cite{gu2021structure,SUBRAMANIAN2021105679}, modifies the objective after each solution is found so that subsequent solutions are distinct from those already identified. Neural deflation adapts this idea to sequentially learn conservation laws while enforcing the functional independence of each newly learned law from those discovered previously. We apply this procedure as the domain is progressively enlarged, using the resulting conservation and deflation losses to determine how many independent conservation laws remain admissible at each stage. 
Their pointwise values over the phase space further help identify where a candidate law begins to fail; see, e.g., Figure~\ref{fig:2d_node}(c)--(d), where a singularity is revealed by a large loss value at the origin. 

The domain size thereby serves as a natural continuation parameter toward
exploring the phase space of the system. For domains centered at a regular point, the method initially seeks the $d-1$ local conservation laws predicted by the flow-box theorem. As the domain expands, the disappearance of one or more of these laws indicates that the corresponding local flow-box structure cannot be extended over the entire region. Conservation laws that persist over substantially larger domains reflect structure beyond that guaranteed by local theory and may correspond to genuine global invariants.

Our numerical experiments show that the proposed framework recovers the local conservation laws predicted by the flow-box theorem (and the associated
Frobenius integrability~\cite{liu2022machine}) and distinguishes them from invariants that persist over larger domains. By tracking successive reductions in the number of valid independent conservation laws, the method differentiates systems and parameter regimes with varying global structures, including integrable, near-integrable, and non-integrable dynamics. It also reveals how the spatial scale over which a conservation law remains valid changes as a system approaches or departs from an integrable limit.

The remainder of the paper is organized as follows.
Section~\ref{sec:background} introduces local and global conservation
laws, functional independence, and the flow-box theorem, and then
formulates the local-to-global continuation problem.
Section~\ref{sec:method} presents neural deflation and explains how it is applied over a nested family of expanding
domains, with the associated losses serving as indicators of the
persistence and breakdown of local conservation laws (possibly in the form
of singularities). Section~\ref{sec:experiments} reports numerical experiments on
elementary dynamical systems and Hamiltonian benchmarks in integrable,
near-integrable, and non-integrable regimes. The paper concludes with a
discussion of the capabilities, limitations, and possible extensions of
the proposed framework.

\section{\label{sec:background}Background and motivation}

\subsection{Local and global conservation laws and their functional
independence}

Consider a $d$-dimensional dynamical system, not necessarily Hamiltonian, on an open phase space $\Omega\subset \R^d$
\begin{align}
\label{eq:dynmiacal_system}
    \frac{d\bx}{dt} = \mathbf{f}(\bx), \quad \mathbf{f}: 
    \Omega \to \R^d.
\end{align}
A function $I: U\to\R$, defined on an open subset
$ U\subset \Omega$, is called a \emph{local conservation law} if it is
constant along every trajectory segment contained in $ U$. More precisely, for any solution $\bx(t)$ of Eq.~\eqref{eq:dynmiacal_system} and any time interval $J$ such that $\bx(t)\in U$ for all $t\in J$,
\begin{align}
    I(\bx(t))\equiv \text{constant}, \qquad t\in J.
\end{align}
If $I$ is continuously differentiable, this condition is equivalent, by the chain rule, to the vanishing of its Lie derivative along $\mathbf{f}$:
\begin{align}
\label{eq:conservation_condition}
    \mathcal{L}_{\mathbf{f}}I(\bx) \coloneqq \mathbf{f}(\bx)\cdot \nabla I(\bx)= 0, \qquad \forall\,\bx\in U.
\end{align}
When $ U = \Omega$, $I$ is called a \textit{global conservation law}.

A collection of local conservation laws, $ I_1,\ldots,I_k: U\to\mathbb{R}$,
is said to be \textit{functionally independent} on $ U$ if the map
\[
    \mathbf{I}
    \coloneqq
    (I_1,\ldots,I_k): U\to\R^k
\]
has full rank almost everywhere on $ U$, namely,
\begin{align}
\label{eq:full_rank}
    \operatorname{rank}D\mathbf{I}(\bx)
    =
    k,
    \qquad
    \text{a.e.}~\bx\in U,
\end{align}
where 
$$D\mathbf{I}(\mathbf{x}) = \left[\partial I_i / \partial x_j\right]_{1\le i\le k,\,1\le j\le d}$$
is the $k \times d$ Jacobian of $\mathbf{I}$ at $\mathbf{x}$.
Equivalently, the gradients
$\nabla I_1(\bx),\ldots,\nabla I_k(\bx)$ are linearly independent at (almost) every point of $ U$.

A point $\bx\in U$ is called a \emph{regular point} of the vector
field if $\mathbf{f}(\bx)\neq 0$. By
Eq.~\eqref{eq:conservation_condition}, the gradient of every conservation
law at a regular point is orthogonal to $\mathbf{f}(\bx)$. Therefore,
at most $d-1$ conservation laws can be functionally independent in a
neighborhood of a regular point.

\subsection{The flow-box theorem and local conservation laws}

The following classical result shows that the upper bound of $d-1$ is always attained
locally near a regular point.

\begin{theorem}[Flow-box theorem~\cite{lee2003smooth}]\label{thm:flow_box}Let $\Omega\subset\R^d$ be open, let
$\mathbf{f}\in C^1(\Omega;\R^d)$, and suppose that
$\mathbf{f}(\bx_0)\neq 0$ for some $\bx_0\in \Omega$. Then there exist
an open neighborhood $ U\subset \Omega$ of $\bx_0$, an open set
$V\subset\R^d$, and a $C^1$-diffeomorphism
\[
    \Phi: U\to V,
    \qquad
    \Phi(\bx)
    =
    \bigl(y_1(\bx),\ldots,y_d(\bx)\bigr),
\]
such that
\begin{align}
\label{eq:straightened_vector_field}
    D\Phi(\bx)\mathbf{f}(\bx)
    =
    \mathbf{e}_1,
    \qquad
    \bx\in U,
\end{align}
where $\mathbf{e}_1=(1,0,\ldots,0)^\top$. Equivalently, in the
coordinates $(y_1,\ldots,y_d)$, the system
\eqref{eq:dynmiacal_system} takes the form
\begin{align}
\label{eq:flow_box_coordinates}
    \frac{dy_1}{dt}=1,
    \qquad
    \frac{dy_i}{dt}=0,
    \quad
    i=2,\ldots,d.
\end{align}
\end{theorem}

The flow-box theorem may be viewed as the rank-one case of the
Frobenius theorem \cite{lee2003smooth}: wherever
$\mathbf{f}\neq 0$, the vector field spans a one-dimensional
distribution; since every one-dimensional distribution is
automatically involutive, the Frobenius theorem implies that it is
locally integrable. The flow-box theorem gives the corresponding
coordinate representation, i.e., the vector field is straightened into
$\partial/\partial y_1$ and the trajectories become coordinate
lines in the $y_1$-direction.

The implication for conservation laws follows almost trivially. Define
\begin{align}
\label{eq:flow_box_conservation_laws}
    I_i(\bx)
    \coloneqq
    y_{i+1}(\bx),
    \qquad
    i=1,\ldots,d-1.
\end{align}
Equation~\eqref{eq:flow_box_coordinates} gives
\[
    \frac{d}{dt}I_i(\bx(t))
    =
    \frac{dy_{i+1}}{dt}
    =
    0,
\]
so each $I_i$ is a local conservation law on $ U$. Moreover,
because $\Phi$ is a diffeomorphism, the gradients
\[
    \nabla I_1(\bx),\ldots,\nabla I_{d-1}(\bx)
\]
are linearly independent throughout $ U$. Thus, every
$d$-dimensional dynamical system admits $d-1$ functionally
independent conservation laws in a sufficiently small neighborhood of
any regular point.

Note that the existence of
$d-1$ conservation laws on a small neighborhood does not imply that
the system is globally integrable or that it possesses any nontrivial
global conservation law. These local invariants simply arise from the
universal geometry of a nonvanishing vector field and such local
features cannot be used to indicate any further special global structure of the dynamical system.

\subsection{The local-to-global continuation problem}

Consequently, the flow-box theorem resolves the local existence question,  yet
does not determine how far the resulting conservation laws can be
continued. For various reasons, a local invariant may fail to extend smoothly to a larger
domain. Moreover, even when several conservation
laws can each be continued, their gradients may cease to remain
linearly independent throughout the enlarged domain.

To formulate this problem (namely the one of continuation of
conservation laws from local to global), let $\{ U_r\}_{r\in\mathcal{R}}$ be a nested family of connected open subsets of $\Omega$, indexed by a
scale parameter $r$, such that
\begin{align}
\label{eq:nested_domains}
     U_{r_1}\subseteq U_{r_2}
    \qquad
    \text{whenever }r_1\leq r_2.
\end{align}
The domains may, for example, be expanding rectangular boxes or, in a
Hamiltonian system, energy sublevel sets.

For each such domain, let $\nu( U)$ denote the maximum number of functionally independent $C^1$ conservation laws defined on $ U$. By definition, $\nu( U)\leq d-1$. Moreover, since the domains are nested,
\begin{align}
\label{eq:nu_monotone}
\nu( U_{r_2})
\leq
\nu( U_{r_1}),
\qquad
r_1\leq r_2.
\end{align}

If the domains are centered at a regular point $\bx_0$, the flow-box
theorem gives
\[
    \nu( U_r)=d-1
\]
for all sufficiently small $r$. As $r$ increases, the value of
$\nu( U_r)$ may decrease. Each decrease marks the ``failure'' of at
least one of the initially available local conservation laws to extend
over the enlarged domain.

Because a family of conservation laws admits nonunique representations, the intrinsic quantity of interest is the number $\nu(U_r)$ of functionally independent conservation laws, rather than any particular choice of representatives. Indeed, if
\[
    \mathbf{I}
    =
    (I_1,\ldots,I_k)
\]
is a functionally independent family and
$F:\R^k\to\R^k$ is a local diffeomorphism, then
$F\circ\mathbf{I}$ defines another family with the same level-set
structure. Consequently, the loss of a conservation law as
the domain expands refers to a decrease in the
number of functionally independent conservation laws,
rather than the disappearance of a particular representative.

The local-to-global continuation problem considered in this paper consists
therefore of determining, or approximating, the function
\(
    r\mapsto\nu( U_r),
\)
and of investigating the mechanisms responsible for its decreases. In particular, after fixing a reference point in phase space, we seek to answer the following questions:
\begin{enumerate}
    \item How many functionally independent conservation laws exist on
    each domain $ U_r$?

    \item At what domain sizes does this number decrease, and can these
    critical domain sizes be interpreted in terms of the underlying
    system dynamics?

    \item Can the region in which a conservation law first fails be
    identified from its pointwise conservation or independence
    residual?

    \item How many and which of the conservation laws available locally persist over
    the full phase-space region of interest and therefore represent
    global structure?
\end{enumerate}

Our primary construction begins with domains centered at a regular
point, where the flow-box theorem supplies a precise local benchmark.
We also consider domains centered at equilibria, particularly energy
sublevel sets in Hamiltonian systems. In that setting, the flow-box
theorem does not apply at the center, and the number of conservation
laws present on small domains depends on additional structure of the
system.

Determining $\nu( U_r)$ analytically is generally difficult,
especially when the vector field is available only through data or
when the loss of a conservation law is caused by complicated global
dynamics or multiple singularities. This motivates a
data-driven approach to the local-to-global continuation problem: To reiterate our approach in context, we first review existing methods for conservation-law discovery, with particular attention to their ability to recover multiple functionally independent invariants and to the phase-space domains on which they operate.

\subsection{\label{sec:related_work}Related work}

Before delving into the techniques utilized in the present work, it
is worthwhile to acknowledge that
a growing literature has developed data-driven methods for discovering
conservation laws from trajectory observations, vector-field
evaluations, or known differential equations. Early sparse-regression
approaches searched for conserved quantities within prescribed
function libraries~\cite{kaiser2018discovering}. Wetzel
et al.~\cite{wetzel2020discovering} used interpretable Siamese neural
networks to identify symmetry invariants and conserved quantities from
pairs of related observations. Liu and
Tegmark~\cite{liu2021machine} instead exploited the geometry and
intrinsic dimension of trajectory data to infer the number of conserved
quantities. Ha and Jeong~\cite{ha2021discovering} trained a neural
network on grouped observations sharing the same invariant values,
while Arora et al.~\cite{arora2023model} proposed a model-agnostic
approach for jointly learning first integrals and the underlying
dynamics from trajectories.

Several methods formulate conservation-law discovery geometrically.
Lu et al.~\cite{lu2023discovering} combined optimal transport with
manifold learning to identify both the number and values of conserved
quantities without requiring an explicit dynamical model or accurate
time information. More recently, Belyshev
et al.~\cite{belyshev2024beyond} combined representation learning with
topological analysis to characterize spaces of conservation laws in
classical and quantum examples. Oellerich and
Emelianenko~\cite{oellerich2024towards} developed a limited-data
framework based on the stability of singular vectors to recover the
number and symbolic form of conservation laws, with an emphasis on
biological and chemical systems. Mebratie
et al.~\cite{mebratie2025machine} introduced an alternative
trajectory-based formulation using indeterminate kernel ridge
regression. Their kernel construction reduces data and computational
requirements and directly produces explicit expressions, although the
main formulation targets a single conservation law.

When the governing differential equations are explicitly available, the
conservation condition can be imposed directly through the directional
derivative of a candidate invariant. Liu
et al.~\cite{liu2022machine} trained neural representations of multiple
conservation laws and introduced a numerical procedure for assessing
their functional independence. The Sparse Invariant Detector of Liu
et al.~\cite{liu2024interpretable} searches a prescribed basis for
sparse, interpretable invariants and has recovered previously unknown
conservation laws in applications with known governing equations and
candidate features. Neural deflation~\cite{zhu2023machine} takes a
sequential approach: after each conservation law is learned, the
objective for the next candidate is deflated to prevent functional
dependence on the preceding laws. This construction is designed to
recover a complete family of functionally independent conservation
laws and, in its original Hamiltonian formulation, also enforces
pairwise involution. Its trajectory-based extension replaces direct
access to the vector field with a learned Hamiltonian
model~\cite{chen2025data}.

These methods differ in their data requirements, representations, and
ability to recover explicit or multiple independent invariants.
Nevertheless, to the best of our knowledge, they all formulate
conservation-law discovery on a fixed dataset or a prescribed region
of phase space. They do not distinguish conservation laws that exist
only locally from those that extend globally, nor do they study how the
number of independent conservation laws changes as the sampled domain
is expanded. The present work addresses this distinction (and the
concomitant variation of the number of 
conservation laws) by applying
neural deflation over nested domains and tracking the persistence and
breakdown of the learned conservation laws.

\section{\label{sec:method}Neural Deflation over expanding domains}

We now develop a numerical framework for approximating the
local-to-global continuation problem formulated in
Section~\ref{sec:background}. The method builds on neural deflation
\cite{zhu2023machine,chen2025data}, which sequentially learns
functionally independent conservation laws by modifying the loss used
to train each new candidate. In contrast to the
original Hamiltonian formulation, we omit pairwise involution
constraints while retaining functional independence through
the gradient-based deflation factor introduced below. This
factor penalizes candidates whose gradients become linearly
dependent on those of previously learned conservation laws.
Our objective is therefore to recover a maximal family of
\textit{independent} conservation laws on each domain, rather than a
maximal family \textit{in involution}. This distinction allows us to
target the full set of $d-1$ local invariants predicted by the
flow-box theorem, which need not be mutually in involution.
The method applies to a general autonomous dynamical
system~\eqref{eq:dynmiacal_system} without assuming a
Hamiltonian or Poisson structure.

We first describe the discovery of conservation laws on a single
prescribed domain. We then apply this fixed-domain procedure to a
nested family of expanding domains and use the resulting losses to
estimate the function $r\mapsto\nu( U_r)$.

\subsection{Learning a conservation law on a fixed domain}
\label{subsec:fixed_domain_conservation}

Let $ U\subset \Omega$ be a prescribed domain. We assume that the
vector field $\mathbf{f}(\bx)$ can be evaluated at sampled points in
$\Omega$, either directly from the governing equations or through
available vector-field data. We draw $N_{\mathrm{tr}}\in\mathbb{Z}^+$ and $N_{\mathrm{val}}\in\mathbb{Z}^+$ points, forming training and validation
sets $\mathcal{T}_U$ and
$\mathcal{V}_U$, respectively, from a probability distribution (typically uniform) supported on $ U$:
\[
    \mathcal{T}_ U
    =
    \left\{\bx^{(n)}\right\}_{n=1}^{N_{\mathrm{tr}}},
    \qquad
    \mathcal{V}_ U
    =
    \left\{\bx^{(n)}\right\}_{n=1}^{N_{\mathrm{val}}},
\]

A candidate conservation law is represented by a scalar-valued neural
network
\[
    I(\bx;\theta): U\to\R,
\]
where $\theta$ denotes the trainable parameters. Its spatial gradient
$\nabla I(\bx;\theta)$ is computed by automatic differentiation~\cite{baydin2018automatic}. The
conservation condition [Eq.~\eqref{eq:conservation_condition}] is
\[
    \mathbf{f}(\bx)\cdot\nabla I(\bx;\theta)=0.
\]
To remove the arbitrary scaling of both the vector field and the
candidate conservation law, we introduce the normalized vectors
\begin{align}
\label{eq:normalized_vectors}
    \widehat{\mathbf{f}}(\bx)
    \coloneqq
    \frac{\mathbf{f}(\bx)}
         {\|\mathbf{f}(\bx)\|_2},
    \quad 
    \widehat{\nabla I}(\bx;\theta)
    \coloneqq
    \frac{\nabla I(\bx;\theta)}
         {\|\nabla I(\bx;\theta)\|_2},
\end{align}
and define the pointwise
conservation loss as
\begin{align}
\label{eq:pointwise_conservation_loss}
    \ell_{\mathrm{cons}}
    [I;\bx]
    \coloneqq
    \left|
        \widehat{\mathbf{f}}(\bx)
        \cdot
        \widehat{\nabla I}(\bx)
    \right|^2.
\end{align}
Equivalently, away from equilibria and critical points of $I$,
\begin{align}
    \ell_{\mathrm{cons}}
    [I;\bx]
    =
    \frac{
        \left|
            \mathbf{f}(\bx)\cdot\nabla I(\bx)
        \right|^2
    }{
        \|\mathbf{f}(\bx)\|_2^2
        \|\nabla I(\bx)\|_2^2
    }.
\end{align}
Thus, $\ell_{\mathrm{cons}}[I;\bx]$ is the squared cosine of the angle between the vector field and the gradient of $I$. It vanishes precisely when the vector field is tangent to the level set of $I$ at $\bx$. The normalization makes the loss invariant under nonzero rescalings of both $\mathbf{f}$ and $I$ and, in particular, prevents the loss from being minimized merely by shrinking the magnitude of the network output.

The first candidate conservation law $I_1(\cdot;\theta_1)$ is trained
by minimizing the empirical conservation loss
\begin{align}
\label{eq:first_conservation_loss}
    \mathcal{L}_{1}
    (\theta_1;\mathcal{T}_ U)
    \coloneqq
    \frac{1}{|\mathcal{T}_ U|}
    \sum_{\bx\in\mathcal{T}_ U}
    \ell_{\mathrm{cons}}
    [I_1(\cdot;\theta_1);\bx].
\end{align}
Let $\theta_1^\ast$ denote the trained parameters. A small validation
loss indicates that $I_1(\cdot;\theta_1^\ast)$ approximately satisfies
the conservation equation over $ U$. This loss alone, however,
cannot be used to learn a complete family of conserved quantities since after one conservation law
has been identified, subsequent networks may simply learn functions of
the same invariant.

\subsection{Functional independence through neural deflation}
\label{subsec:neural_deflation}

To prevent such redundancy, we employ neural deflation, which
penalizes the functional dependence of each new candidate on
previously learned conservation laws. Specifically, suppose
that $K-1$ candidate conservation laws
\[
    I_1(\cdot;\theta_1^\ast),
    \ldots,
    I_{K-1}(\cdot;\theta_{K-1}^\ast)
\]
have already been learned on a fixed domain $U$. We then train
the next candidate $I_K(\cdot;\theta_K)$ by optimizing only
$\theta_K$, while keeping the parameters
$\theta_1^\ast,\ldots,\theta_{K-1}^\ast$ fixed. The previously
learned functions therefore remain unchanged during this
training stage, and their spatial gradients are used to
assess the functional independence of the new candidate.

At each point $\bx\in U$, define the linear subspace spanned by the
normalized gradients previously learned:
\begin{align}
\label{eq:previous_gradient_span}
    \mathcal{S}_{K-1}(\bx)
    \coloneqq
    \operatorname{span}
    \left\{
        \widehat{\nabla I_1}
        (\bx;\theta_1^\ast),
        \ldots,
        \widehat{\nabla I_{K-1}}
        (\bx;\theta_{K-1}^\ast)
    \right\}
    \subset\R^d.
\end{align}
The component of the new gradient that is orthogonal to this subspace
is
\begin{align}
\label{eq:orthogonal_gradient_component}
    \mathbf{p}_K(\bx;\theta_K)
    \coloneqq
    \operatorname{Proj}_{\mathcal{S}_{K-1}(\bx)^\perp}
    \widehat{\nabla I_K}(\bx;\theta_K),
\end{align}
where $\operatorname{Proj}_V:\R^d\to\R^d$ denotes the projection onto the subspace $V\subseteq\mathbb{R}^d$.
We use the squared norm $\|\mathbf{p}_K\|$ as a pointwise independence factor:
\begin{align}
\label{eq:independence_factor}
    \ell_{\mathrm{ind},K}
    [I_K(\cdot; \theta_K); \bx]
    \coloneqq
    \left\|
        \mathbf{p}_K(\bx;\theta_K)
    \right\|_2^2.
\end{align}
Because the candidate gradient is normalized, it's clear that
$0\leq\ell_{\mathrm{ind},K}\leq 1$. Moreover,
\[
    \ell_{\mathrm{ind},K}(\bx;\theta_K)=0
\]
if and only if the gradient of $I_K$ lies in the span of the previous
gradients at $\bx$. In particular, if $I_K$ is locally a smooth
function of the previously learned conservation laws, then its
independence factor vanishes.

For $K\geq 2$, define the pointwise deflated loss by
\begin{align}
\label{eq:pointwise_deflated_loss}
    \ell_{\mathrm{defl},K}
    \left[I_K(\cdot; \theta_K); \bx\right]
    \coloneqq
    \frac{
        \ell_{\mathrm{cons}}
        [I_K(\cdot;\theta_K);\bx]
    }{
        \left|
            \ell_{\mathrm{ind},K}\left[I_K(\cdot; \theta_K); \bx\right]
        \right|^\alpha
    },
\end{align}
where $\alpha>0$ is a tunable deflation parameter. The numerator
enforces the conservation condition, whereas the denominator penalizes
functional dependence on the previously learned conservation laws.
In particular, $\ell_{\mathrm{ind},K}\left[I_K(\cdot; \theta_K); \bx\right]$ vanishes when
$\nabla I_K(\bx;\theta_K)$ lies in the span of the preceding
gradients, causing the pointwise deflated loss to become singular.
This prevents the new candidate from reproducing a nonlinear function
of the previously discovered conservation laws. The exponent $\alpha$
controls the strength of the penalty near such dependent
configurations.

The $K$th network is trained by minimizing the empirical deflated loss
over the training set:
\begin{align}
\label{eq:empirical_deflated_loss}
    \mathcal{L}_{K}
    (\theta_K;\mathcal{T}_U)
    \coloneqq
    \frac{1}{|\mathcal{T}_U|}
    \sum_{\bx\in\mathcal{T}_U}
    \ell_{\mathrm{defl},K}
    \left[I_K(\cdot; \theta_K); \bx\right].
\end{align}
In the original neural deflation method for Hamiltonian systems, the
numerator also contains pairwise involution losses between the new
candidate and the previously learned conservation laws
\cite{zhu2023machine}. We omit these terms here. For a general vector
field, there need not be a canonical Poisson bracket, and our objective
is to recover a maximal family of functionally independent
conservation laws rather than a family of first integrals in
involution.

As in the original neural deflation framework~\cite{zhu2023machine}, a notable property of the proposed loss is its
consistency in the infinite-sample limit. In this limit, the empirical
average over the training set is replaced by an expectation with
respect to a probability measure on the domain $U$. The following theorem
summarizes this property.

\begin{theorem}[Consistency of neural deflation]
\label{thm:deflation_consistency}
Assume that $\{I_k(\cdot;\theta_k^*)\}_{k=1}^{K-1}$ is a set of $K-1$ functionally independent conservation laws on
$U\subset\R^d$, where $\theta_k^*\in\mathbb{R}^p$ denotes the $p\in\mathbb{Z}^+$ learned parameters after training. Let $\mu$ be a probability measure on $U$ that is
mutually absolutely continuous with respect to the Lebesgue measure on $U$.
For a candidate $I_K(\cdot;\theta_K)$, define the population deflated
loss
\begin{align}
\label{eq:population_deflated_loss}
    \mathcal{L}_{K}^{\mu}(\theta_K)
    \coloneqq
    \mathbb{E}_{\bx\sim\mu}
    \left[
        \frac{
            \ell_{\mathrm{cons}}
            [I_K(\cdot;\theta_K);\bx]
        }{
            \left|
                \ell_{\mathrm{ind},K}\left[I_K(\cdot; \theta_K); \bx\right]
            \right|^\alpha
        }
    \right],
\end{align}
where $\alpha>0$ and the integrand is interpreted as infinite whenever
$\ell_{\mathrm{ind},K}\left[I_K(\cdot; \theta_K); \bx\right]=0$. Then $\mathcal{L}_{K}^{\mu}(\theta_K^*)=0$ 
if and only if
\begin{itemize}
    \item $I_K(\cdot;\theta_K^*)$ is a conservation law, and
    \item $\{I_k(\cdot;\theta_k^*)\}_{k=1}^{K}$ is functionally independent $\mu$-almost everywhere on $U$.
\end{itemize}
\end{theorem}
The proof is a straightforward adaptation of the argument for the original neural deflation method and is deferred to Appendix~\ref{app:proof} for completeness.

\subsection{Estimating the number of conservation laws}
\label{subsec:estimating_number}

For a fixed domain $ U$, we train the networks sequentially: $I_1,\ I_2,\ \ldots, I_{d-1}$. Because a $d$-dimensional system admits at most $d-1$ functionally
independent (local) conservation laws, it is sufficient
to consider at most $d-1$ networks. After training the $K$th network,
we evaluate its deflated loss and conservation loss on the independent validation set:
\begin{align}
\label{eq:validation_deflated_loss}
    \mathcal{L}_{\mathrm{defl},K}^{\mathrm{val}}( U)
    &\coloneqq
    \frac{1}{|\mathcal{V}_ U|}
    \sum_{\bx\in\mathcal{V}_ U}
    \ell_{\mathrm{defl},K}
    \left[I_K(\cdot; \theta_K^\ast); \bx\right].\\
\label{eq:validation_conservation_loss}
    \mathcal{L}_{\mathrm{cons}, K}^{\mathrm{val}}( U)
    &\coloneqq
    \frac{1}{|\mathcal{V}_ U|}
    \sum_{\bx\in\mathcal{V}_ U}
    \ell_{\mathrm{cons}}
    [I_K(\cdot; \theta_K^\ast); \bx].
\end{align}
For $K=1$, the deflated loss is understood to be the conservation loss
in \eqref{eq:first_conservation_loss}.

If $ U$ supports at least $K$ independent conservation laws, the
first $K$ validation losses are expected to remain near the baseline
optimization and approximation error. Once the method attempts to
learn more independent conservation laws than the domain supports, the
conservation and independence requirements become incompatible, and a
pronounced increase in the validation loss is expected. We therefore
estimate the number of independent conservation laws by the index
preceding the first significant loss increase:
\begin{align}
\label{eq:estimated_number}
    \widehat{\nu}( U)
    =
    K-1
    ~
    \text{if a significant increase first occurs at index $K$.}
\end{align}
Rather than assuming a universal absolute threshold, we examine both
the magnitude of the losses and their separation across successive
values of $K$. These will be illustrated along with numerical experiments in  Section~\ref{sec:experiments}.

\subsection{Continuation over expanding domains}
\label{subsec:expanding_domains}

We now apply the neural deflation procedure to a nested
family
\[
     U_{r_1}
    \subset
     U_{r_2}
    \subset
    \cdots
    \subset
     U_{r_M},
    \qquad
    r_1<r_2<\cdots<r_M.
\]
For each $r_m$, separate training and validation samples are drawn from
$ U_{r_m}$, and the sequence of neural conservation laws is learned
on that domain. This produces an estimate $\widehat{\nu}( U_{r_m})$
for every domain size. Typical choices include
\begin{itemize}
    \item rectangular domains centered at a prescribed
point $\bx_0$,
\begin{align}
\label{eq:rectangular_domain}
     U_r
    =
    \bx_0+(-r,r)^d;
\end{align}
\item energy sublevel sets in Hamiltonian examples,
\begin{align}
\label{eq:energy_sublevel_domain}
     U_r
    =
    \{\bx\in D:H(\bx)<r\}.
\end{align}
\end{itemize}
The first construction allows the domain to grow from either a regular
point or an equilibrium. The second respects the geometry of the
Hamiltonian and is useful when the continuation is naturally
parameterized by energy.

At every fixed value of $r$, the method learns conservation laws only
on $ U_r$. These learned quantities should therefore be regarded
as local conservation laws on that domain unless $ U_r$ covers the
entire phase-space region under consideration. The local-to-global
information is obtained from
the variation of $r\mapsto\widehat{\nu}( U_r)$ over the nested family.

When the domains shrink to a regular point as $r\to 0$, Theorem~\ref{thm:flow_box} predicts
\[
    \widehat{\nu}( U_r)\approx d-1
\]
for sufficiently small $r$. A decrease in
$\widehat{\nu}( U_r)$ as $r$ grows indicates that at least one of
the conservation laws supported on smaller domains cannot be extended
as part of an independent family over the enlarged region. Although numerical persistence alone does
not provide a mathematically rigorous statement about the global validity of a conservation law, it does provide evidence that a candidate local law, which
continues to be approximately recovered over all tested domains, {\it may} be a global law.

\subsection{Average, maximal, and pointwise diagnostics}
\label{subsec:loss_diagnostics}

The estimated number $\widehat{\nu}( U_r)$ detects whether a loss
of conservation laws has occurred, but it does not by itself reveal
where the obstruction is located and/or possibly its nature. To investigate the spatial structure
of a breakdown, we inspect the pointwise conservation and deflated
losses over the domain.

Specifically, to localize a potential obstruction, we visualize the pointwise fields
\[
    \bx
    \mapsto
    \ell_{\mathrm{cons}}[I_{K,r}(\cdot; \theta_{K, r}^\ast);\bx],
    \qquad
    \bx
    \mapsto
    \ell_{\mathrm{defl},K}
    [I_{K,r}(\cdot; \theta_{K, r}^\ast);\bx],
\]
either throughout the domain $U_r$ or on selected lower-dimensional slices.
Regions of large conservation loss indicate that the candidate
function fails to remain constant along the vector field. Regions of
small independence factor indicate that its gradient becomes nearly
dependent on the previously learned gradients. The deflated loss
combines these two effects, so it is sometimes useful to inspect each quantity separately when interpreting a breakdown.

In some cases the interpretation of breakdown is simple since the high-loss region may concentrate near an
equilibrium or another isolated singularity, thereby localizing the
obstruction to continuation.  In more complicated scenarios (e.g., when a recurrence arises 
along an oscillatory trajectory), the method may still sharply identify the critical value of $r$ at which a conservation law is lost, while failing to resolve the specific location on the corresponding ``$r$-shell'' where the breakdown occurs.

\subsection{Summary of the expanding-domain procedure}
\label{subsec:method_summary}

The complete procedure can be summarized as follows:
\begin{enumerate}
    \item Choose a nested family of domains
    $\{ U_{r_m}\}_{m=1}^M$.

    \item For each domain $ U_{r_m}$, draw independent training and
    validation samples and evaluate the vector field at these points.

    \item Train the first neural conservation law using the normalized
    conservation loss \eqref{eq:first_conservation_loss}.

    \item Holding all previously trained networks fixed, sequentially
    train additional candidates using the deflated loss
    \eqref{eq:empirical_deflated_loss}.

    \item Identify the first candidate whose validation loss exhibits
    a pronounced increase, and use the preceding index to estimate
    $\widehat{\nu}( U_{r_m})$.

    \item Compare the results across domain sizes to detect decreases
    in $\widehat{\nu}( U_r)$ and, whenever possible, localize the
    corresponding obstruction.
\end{enumerate}


\section{\label{sec:experiments}Numerical experiments}
We now apply the expanding-domain neural deflation framework to a
collection of dynamical systems of increasing complexity. We begin
with elementary two- and three-dimensional examples for which the
local and global conservation laws are known analytically. We then
consider the H\'enon--Heiles system in both integrable and
non-integrable regimes, followed by the non-integrable double pendulum system.

For each system, we specify a nested family of training domains
$\{U_r\}$ and apply neural deflation independently on each domain. We
use the validation conservation and deflated losses to estimate the
number of conservation laws supported on $U_r$. When a loss increases,
we additionally examine its pointwise values to determine whether the
corresponding obstruction can be localized/identified within the domain.

\subsection{Network architecture and training setup}
\label{subsec:network_training}

Unless otherwise stated, we use the following common training setup. Each candidate conservation law $I_k(\bx;\theta_k)$ is represented by
a fully connected neural network with four hidden layers of width
$64$. The SiLU activation function~\cite{elfwing2018sigmoid}, $\operatorname{SiLU}(z)=z/(1+e^{-z})$, is applied after each hidden linear layer, while the output layer is linear and scalar-valued. All weight matrices are initialized using Xavier uniform initialization~\cite{glorot2010understanding}, and all
biases are initialized to zero.

For each domain $U_r\subset\R^d$, we independently sample a training
set $\mathcal{T}_{U_r}$ and a validation set
$\mathcal{V}_{U_r}$, each consisting of $50{,}000$ points drawn
uniformly from $U_r$. The network gradients with respect to the phase
variables are evaluated using automatic differentiation. Each
candidate network is trained for $1000$ epochs using the Adam
optimizer~\cite{kingma2015adam} with learning rate $10^{-3}$ and a batch size of $32$. The
deflation exponent is set to $\alpha=1$ in all experiments; however, the results are fairly robust to this choice. For example, $\alpha=0.5$ yields qualitatively similar results.

We report both the conservation loss and the deflated loss on the
validation set. Their averages measure the overall validity of a
candidate law on $U_r$, while their maximal and pointwise values are
used to detect localized failures. Since the training and validation
points are sampled from continuous distributions, isolated
equilibria and other sets of Lebesgue measure zero are not sampled
exactly. The empirical losses therefore assess conservation and
functional independence in an almost-everywhere sense, consistent
with Theorem~\ref{thm:deflation_consistency}.

\subsection{Elementary two- and three-dimensional systems}
\label{subsec:elementary_systems}

We begin with four elementary systems for which the relevant local and
global conservation laws are known analytically. These examples test
whether neural deflation can recover the local conservation laws
predicted by the flow-box theorem, detect their loss as the domain
expands, and distinguish a genuine obstruction to continuation from an
equilibrium across which a global conservation law remains well
defined.

\begin{example}[2D node]
\label{ex:node_2d}
Consider
\begin{align}
\label{eq:experiment_node_2d}
    \left\{
    \begin{aligned}
        \dot{x} &= x,\\
        \dot{y} &= y.
    \end{aligned}
    \right.
\end{align}
This system has no nonconstant global conservation law on $\R^2$.
However, on any simply connected domain
$U\subset\R^2\setminus\{(0,0)\}$, a local conservation law can be
expressed as a smooth function of the polar angle. In particular, on a box
centered at $(2,0)$ that does not reach the origin,
\begin{align}
\label{eq:node_2d_local_law}
    I(x,y)=\frac{y}{x}
\end{align}
is a smooth local conservation law.
\end{example}

\begin{example}[3D degenerate node]
\label{ex:node_3d}
Consider
\begin{align}
\label{eq:experiment_node_3d}
    \left\{
    \begin{aligned}
        \dot{x} &= x,\\
        \dot{y} &= y,\\
        \dot{z} &= 0.
    \end{aligned}
    \right.
\end{align}
The system has the global conservation law
\begin{align}
\label{eq:node_3d_global_law}
    I_1(x,y,z)=z.
\end{align}
Away from the equilibrium line
\[
    \mathcal{E}
    =
    \{(0,0,z):z\in\R\},
\]
it also admits an additional local conservation law. In particular, on
a box centered at $(2,0,0)$ that does not intersect $\mathcal{E}$,
\begin{align}
\label{eq:node_3d_local_law}
    I_2(x,y,z)=\frac{y}{x}
\end{align}
is functionally independent of $I_1$.
\end{example}

\begin{example}[2D center]
\label{ex:center_2d}
Consider
\begin{align}
\label{eq:experiment_center_2d}
    \left\{
    \begin{aligned}
        \dot{x} &= y,\\
        \dot{y} &= -x.
    \end{aligned}
    \right.
\end{align}
The system has the global conservation law (energy)
\begin{align}
\label{eq:center_2d_global_law}
    I(x,y)=x^2+y^2.
\end{align}
Although the origin is an equilibrium and a critical point of $I$, the
conservation law extends smoothly across it.
\end{example}

\begin{example}[3D degenerate center]
\label{ex:center_3d}
Consider
\begin{align}
\label{eq:experiment_center_3d}
    \left\{
    \begin{aligned}
        \dot{x} &= y,\\
        \dot{y} &= -x,\\
        \dot{z} &= 0.
    \end{aligned}
    \right.
\end{align}
This system has two globally defined conservation laws,
\begin{align}
\label{eq:center_3d_global_laws}
    I_1(x,y,z)
    &=
    x^2+y^2,
    &
    I_2(x,y,z)
    &=
    z.
\end{align}
Their gradients are independent away from the equilibrium line
$\mathcal{E}=\{(0,0,z):z\in\R\}$ and hence independent almost
everywhere.
\end{example}

For these examples, let
\[
    \bx_0^{(2)}=(2,0),
    \qquad
    \bx_0^{(3)}=(2,0,0),
\]
and define the expanding boxes
\begin{align}
\label{eq:elementary_domains}
    U_r^{(d)}
    \coloneqq
    \bx_0^{(d)}+[-r,r]^d,
    \qquad
    d\in\{2,3\}.
\end{align}
The two-dimensional systems are trained on $U_r^{(2)}$, while the
three-dimensional systems are trained on $U_r^{(3)}$. We consider
\[
    r\in\{1.5,1.9,2.1,2.5\},
\]
which samples the domains immediately before and after the critical
value $r=2$. At this value, $U_r^{(2)}$ first reaches the origin,
whereas $U_r^{(3)}$ first intersects the equilibrium line
$\mathcal{E}$.

For the 2D node (Example~\ref{ex:node_2d}), the box $U_r^{(2)}$ supports one local conservation
law as long as $r<2$. Figure~\ref{fig:2d_node} shows the pointwise
conservation loss of the learned candidate. For $r=1.5$ and $r=1.9$,
the loss remains small throughout the domain, consistent with the
existence of the local conservation law  $y/x$. Once the box expands past the critical value $r=2$ and it now contains the origin, a pronounced
high-loss region appears near $(0,0)$. This indicates that the local
conservation law cannot be continued over the enlarged domain and
correctly identifies the origin as the obstruction.

\begin{figure}
    \centering
    \begin{subfigure}[t]{0.45\textwidth}
        \centering
        \includegraphics[width=\textwidth]{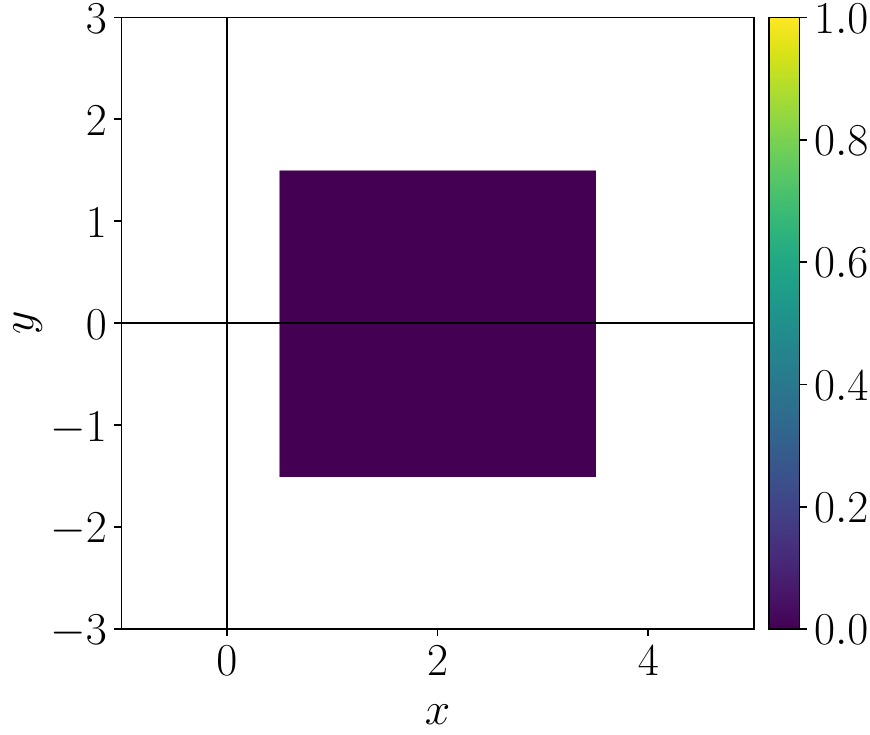}
        \subcaption{$r=1.5$}
    \end{subfigure}
    \hfill
    \begin{subfigure}[t]{0.45\textwidth}
        \centering
        \includegraphics[width=\textwidth]{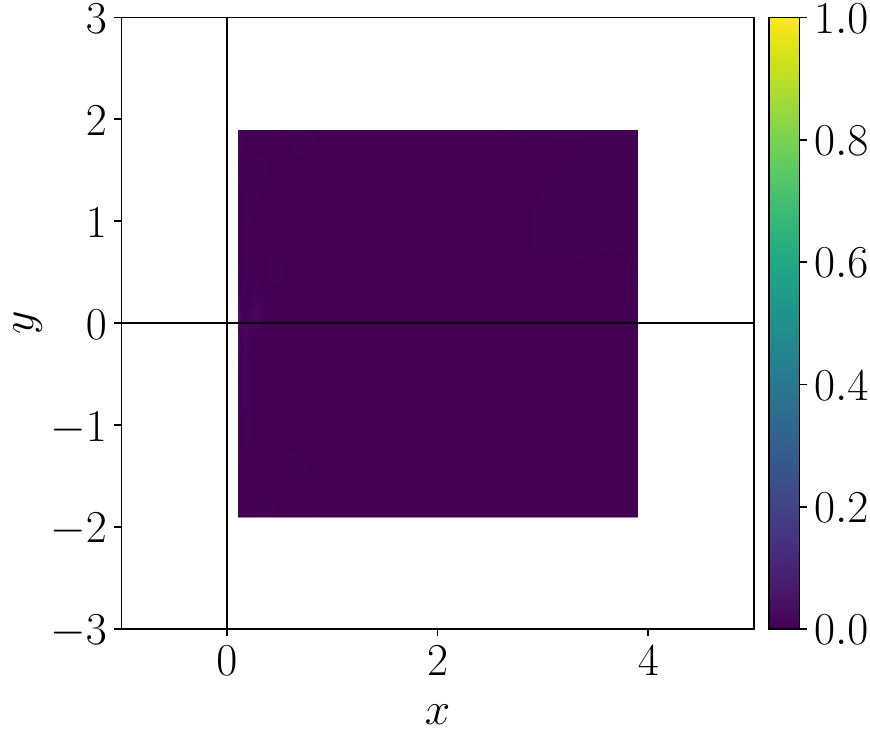}
        \subcaption{$r=1.9$}        
    \end{subfigure}
    \begin{subfigure}[t]{0.45\textwidth}
        \centering
        \includegraphics[width=\textwidth]{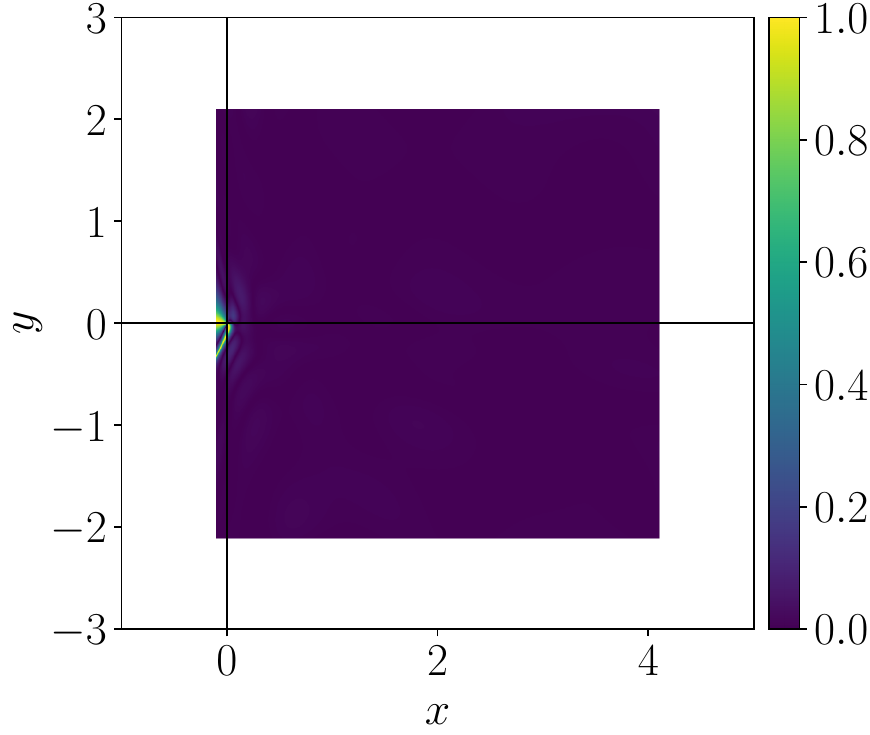}
        \subcaption{$r=2.1$}        
    \end{subfigure}
    \hfill
    \begin{subfigure}[t]{0.45\textwidth}
        \centering
        \includegraphics[width=\textwidth]{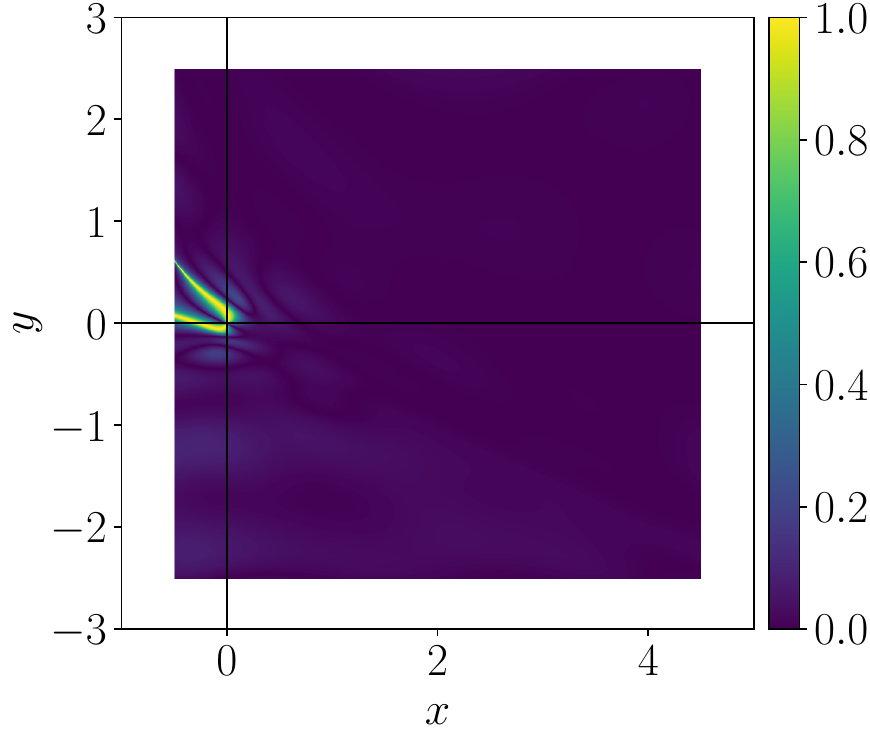}
        \subcaption{$r=2.5$}        
    \end{subfigure}
    \hfill  
    \caption{Pointwise conservation loss for the 2D node (Example~\ref{ex:node_2d}) on
$U_r^{(2)}=[2-r,2+r]\times[-r,r]$, with
$r=1.5,1.9,2.1,2.5$ from left to right. The loss remains small for
$r<2$ but jumps at the origin once $r>2$, where the local
conservation law $y/x$ can no longer be continued smoothly.}
    \label{fig:2d_node}
\end{figure}

For the 3D degenerate node (Example~\ref{ex:node_3d}), the box $U_r^{(3)}$ supports two
independent conservation laws for $r<2$: the global law $z$ and an
additional local law, such as $y/x$. Figure~\ref{fig:3d_node} shows the corresponding pointwise deflated losses on the slice $z=0$. For $r=1.5$
and $r=1.9$, both losses remain small. Once the box expands past
$r=2$ and intersects the equilibrium line, the loss for the first law
remains small, while the deflated loss for the second develops a sharp
increase near $(x,y)=(0,0)$. Thus, the global law persists, whereas
the additional local law breaks down, reducing the number of
independent conservation laws from two to one.

\begin{figure}
    \centering
    \begin{subfigure}[t]{0.24\textwidth}
        \centering
        \includegraphics[width=\textwidth]{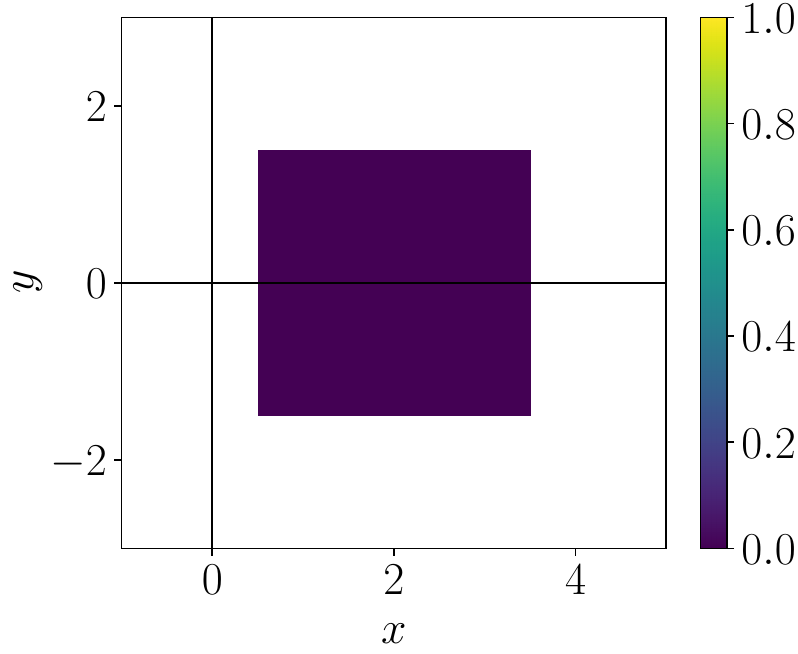}
        \subcaption{$I_1,~r=1.5$}
    \end{subfigure}
    \hfill
    \begin{subfigure}[t]{0.24\textwidth}
        \centering
        \includegraphics[width=\textwidth]{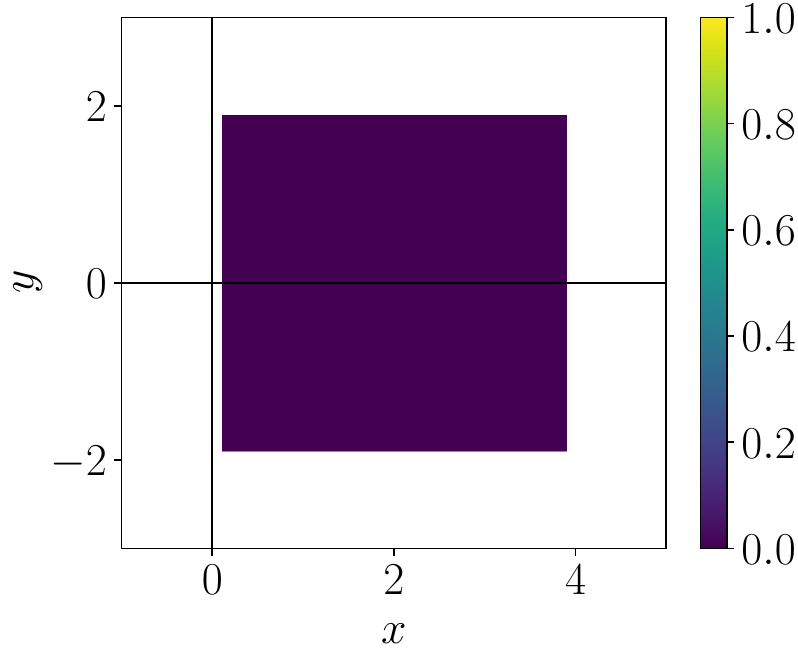}
        \subcaption{$I_1,~r=1.9$}        
    \end{subfigure}
    \begin{subfigure}[t]{0.24\textwidth}
        \centering
        \includegraphics[width=\textwidth]{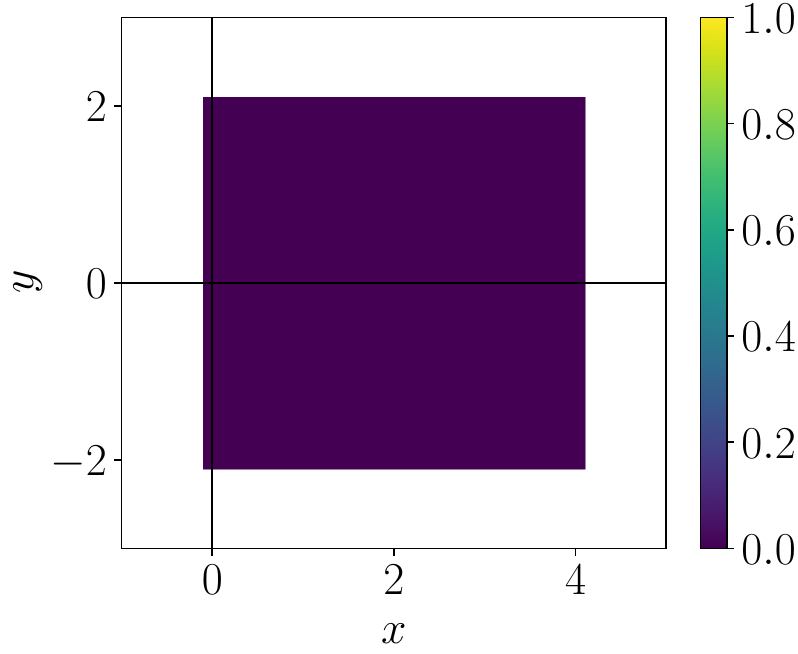}
        \subcaption{$I_1,~r=2.1$}        
    \end{subfigure}
    \hfill
    \begin{subfigure}[t]{0.24\textwidth}
        \centering
        \includegraphics[width=\textwidth]{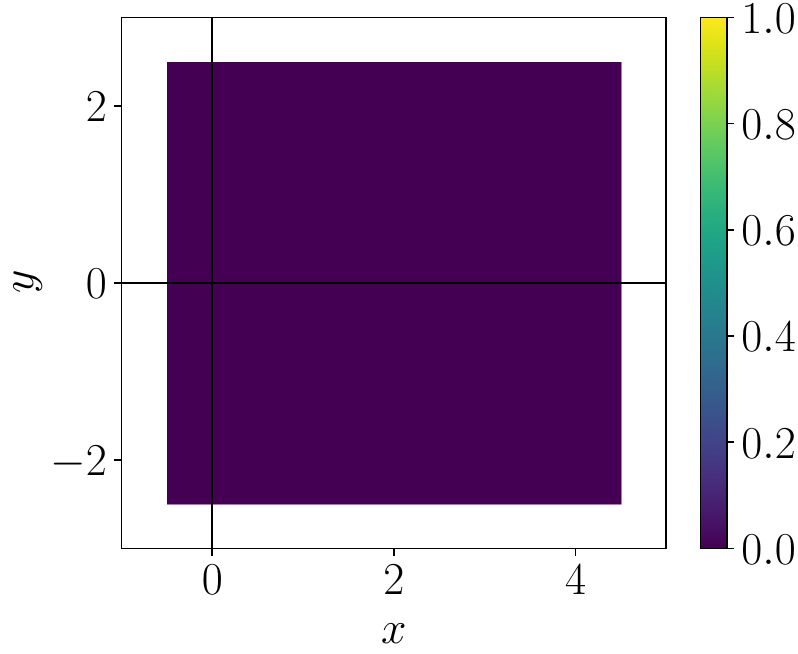}
        \subcaption{$I_1,~r=2.5$}        
    \end{subfigure}
    \hfill
        \begin{subfigure}[t]{0.24\textwidth}
        \centering
        \includegraphics[width=\textwidth]{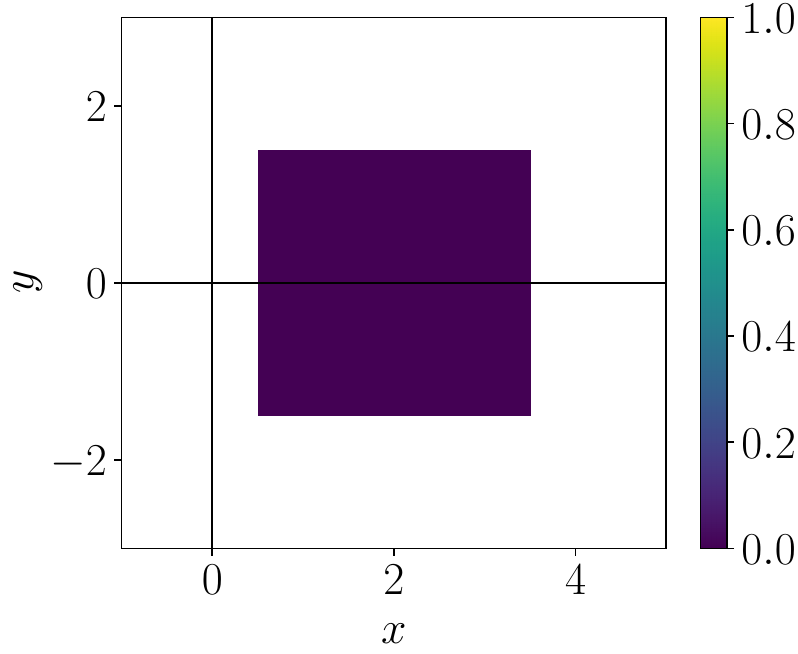}
        \subcaption{$I_2,~r=1.5$}
    \end{subfigure}
    \hfill
    \begin{subfigure}[t]{0.24\textwidth}
        \centering
        \includegraphics[width=\textwidth]{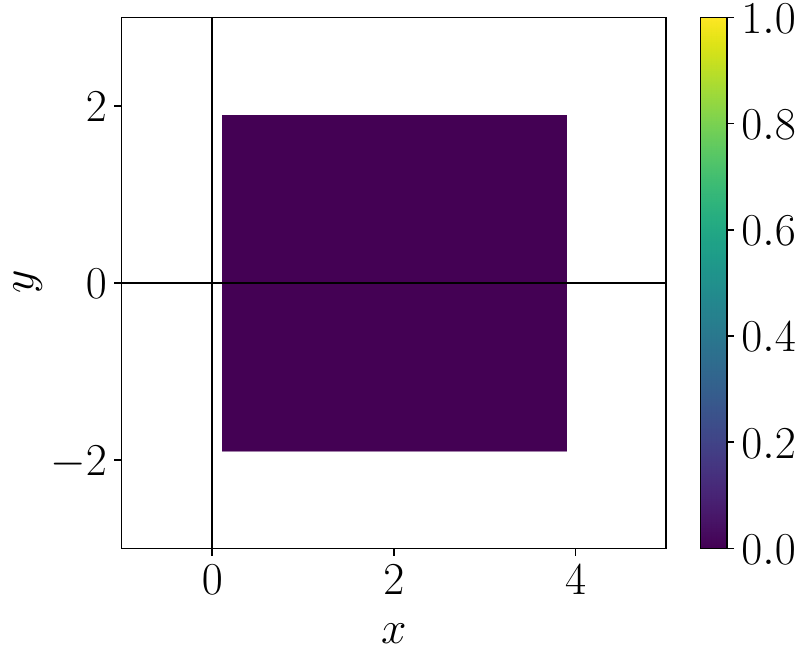}
        \subcaption{$I_2,~r=1.9$}        
    \end{subfigure}
    \begin{subfigure}[t]{0.24\textwidth}
        \centering
        \includegraphics[width=\textwidth]{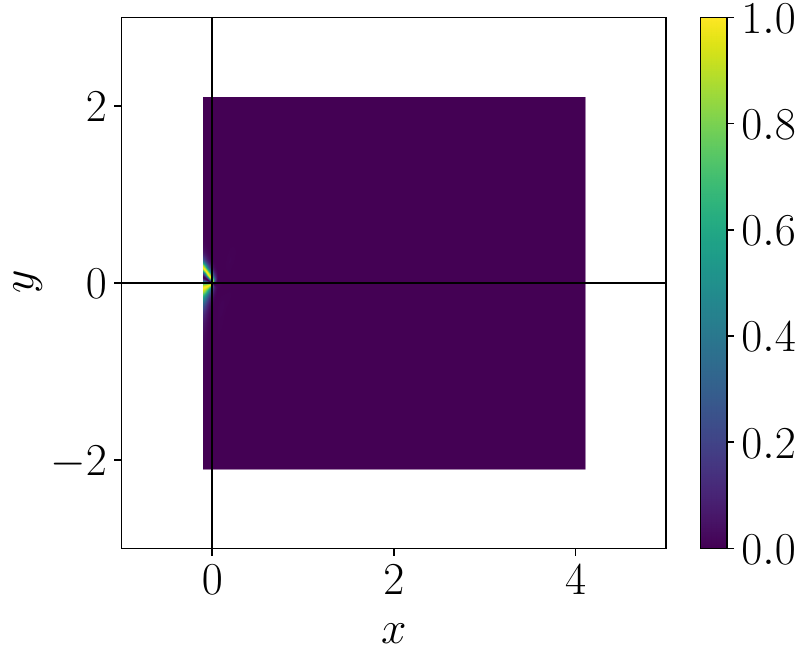}
        \subcaption{$I_2,~r=2.1$}        
    \end{subfigure}
    \hfill
    \begin{subfigure}[t]{0.24\textwidth}
        \centering
        \includegraphics[width=\textwidth]{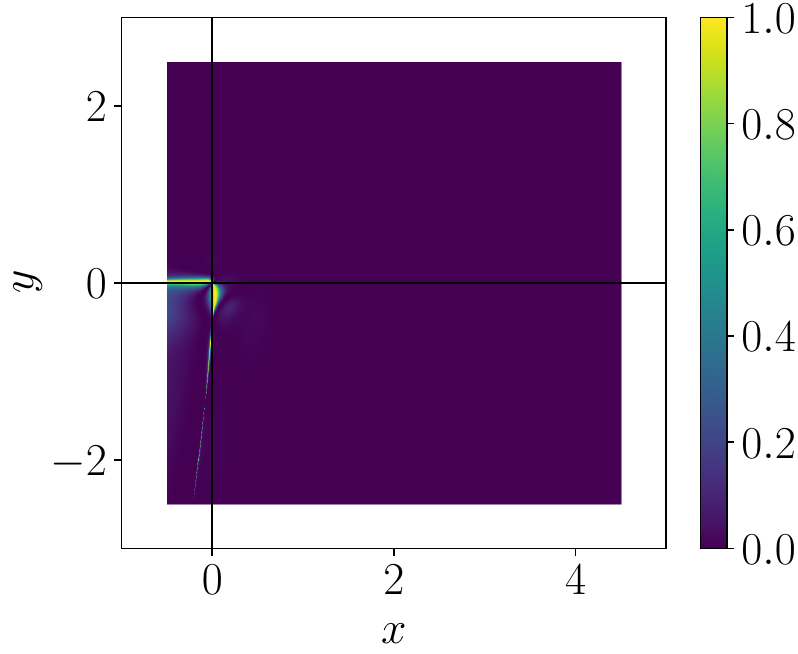}
        \subcaption{$I_2,~r=2.5$}        
    \end{subfigure}
    \caption{Pointwise losses for the 3D degenerate node (Example~\ref{ex:node_3d}) on the slice $z=0$ of
$U_r^{(3)}=[2-r,2+r]\times[-r,r]^2$, with
$r=1.5,1.9,2.1,2.5$ from left to right. Panels (a)--(d) show the
conservation loss, equivalently the deflated loss, for $I_1$, while
panels (e)--(h) show the deflated loss for $I_2$. The first
conservation law persists for all $r$, whereas the loss for the second
jumps near the equilibrium line once $r>2$.}
    \label{fig:3d_node}
\end{figure}


In contrast to the 2D node, the equilibrium at the origin does not
obstruct the conservation law of the 2D center
(Example~\ref{ex:center_2d}). Figure~\ref{fig:2d_center} shows the
pointwise conservation loss of the learned candidate. The loss remains
small for $r=1.5$ and $r=1.9$ and continues to remain small after the
box expands past $r=2$ and contains the origin. This is consistent with
the fact that $x^2+y^2$ is a global conservation law that extends
smoothly across the equilibrium.
\begin{figure}
    \centering
    \begin{subfigure}[t]{0.45\textwidth}
        \centering
        \includegraphics[width=\textwidth]{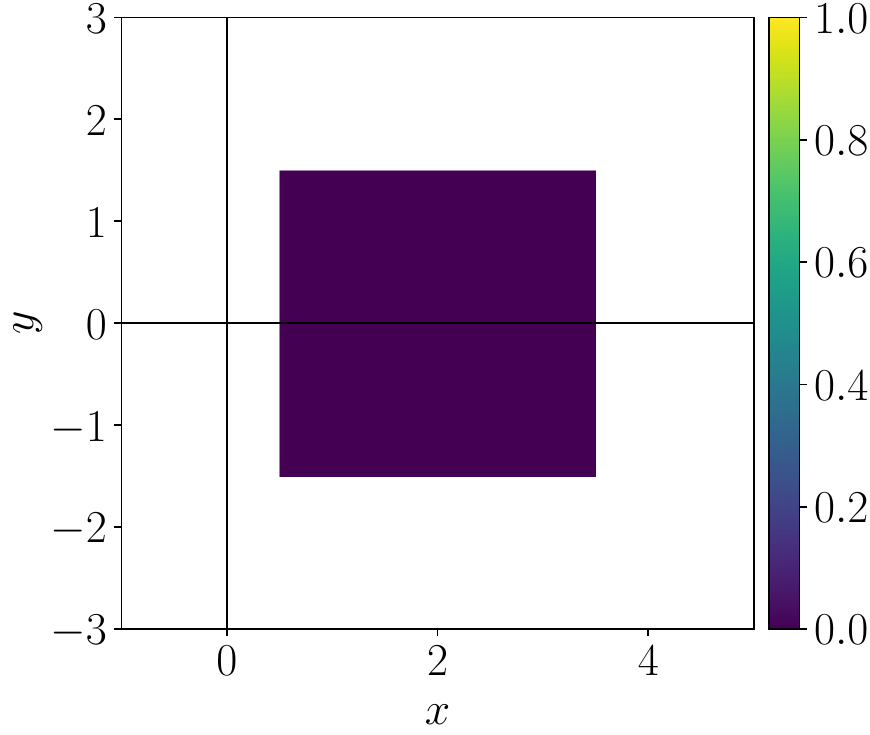}
        \subcaption{$r=1.5$}
    \end{subfigure}
    \hfill
    \begin{subfigure}[t]{0.45\textwidth}
        \centering
        \includegraphics[width=\textwidth]{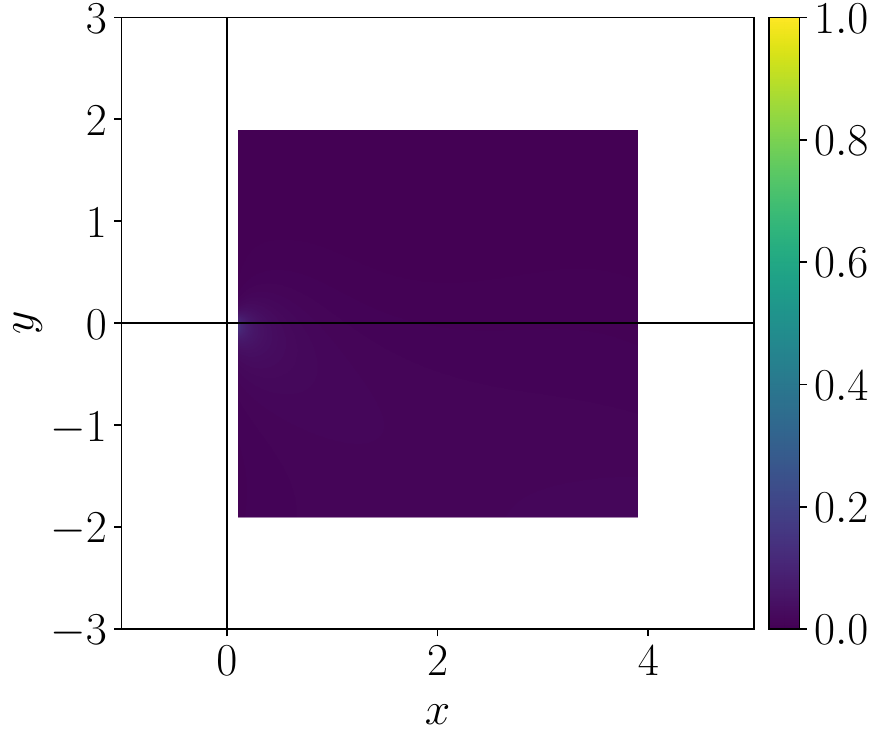}
        \subcaption{$r=1.9$}        
    \end{subfigure}
    \begin{subfigure}[t]{0.45\textwidth}
        \centering
        \includegraphics[width=\textwidth]{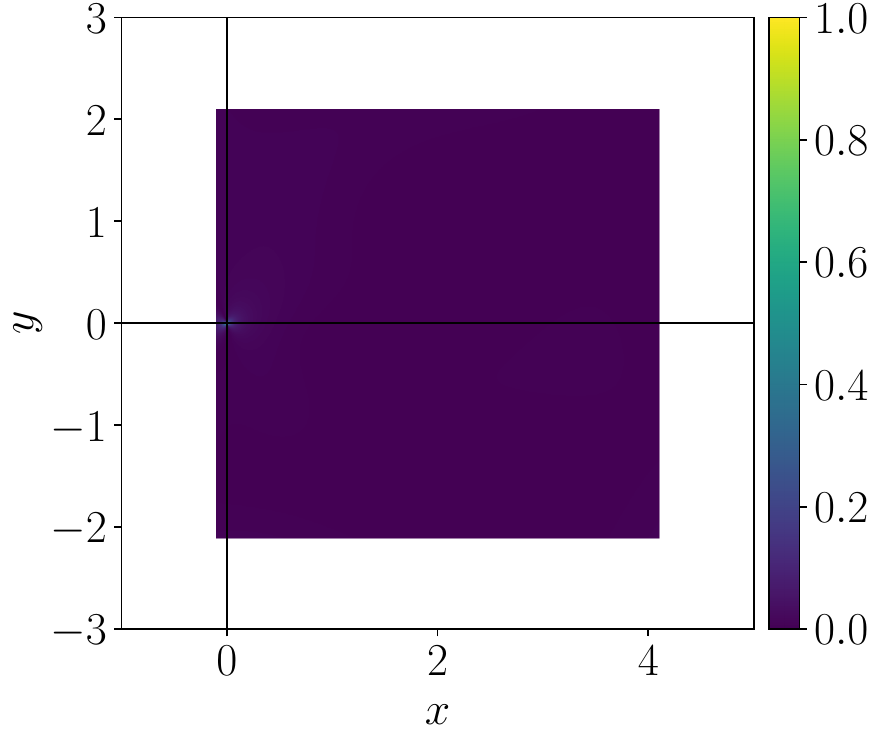}
        \subcaption{$r=2.1$}        
    \end{subfigure}
    \hfill
    \begin{subfigure}[t]{0.45\textwidth}
        \centering
        \includegraphics[width=\textwidth]{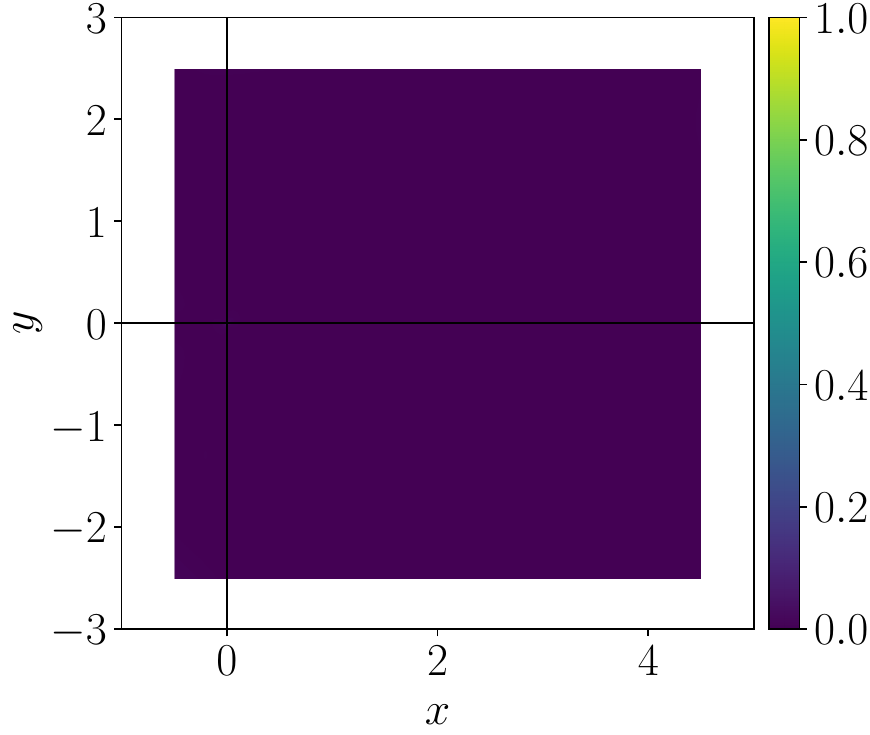}
        \subcaption{$r=2.5$}        
    \end{subfigure}
    \hfill  
    \caption{Pointwise conservation loss for the 2D center
(Example~\ref{ex:center_2d}) on
$U_r^{(2)}=[2-r,2+r]\times[-r,r]$, with
$r=1.5,1.9,2.1,2.5$ from left to right. Unlike the 2D node, the
conservation law $x^2+y^2$ persists after the domain contains the
origin.
}
    \label{fig:2d_center}
\end{figure}

The same distinction appears in three dimensions. For the 3D
degenerate center (Example~\ref{ex:center_3d}), both $x^2+y^2$ and $z$
are globally defined, and neither conservation law is lost when the
box intersects the equilibrium line. Figure~\ref{fig:3d_center} shows
the pointwise deflated losses on the slice $z=0$. For all tested values
of $r$, both learned candidates retain small losses.

\begin{figure}
    \centering
    \begin{subfigure}[t]{0.24\textwidth}
        \centering
        \includegraphics[width=\textwidth]{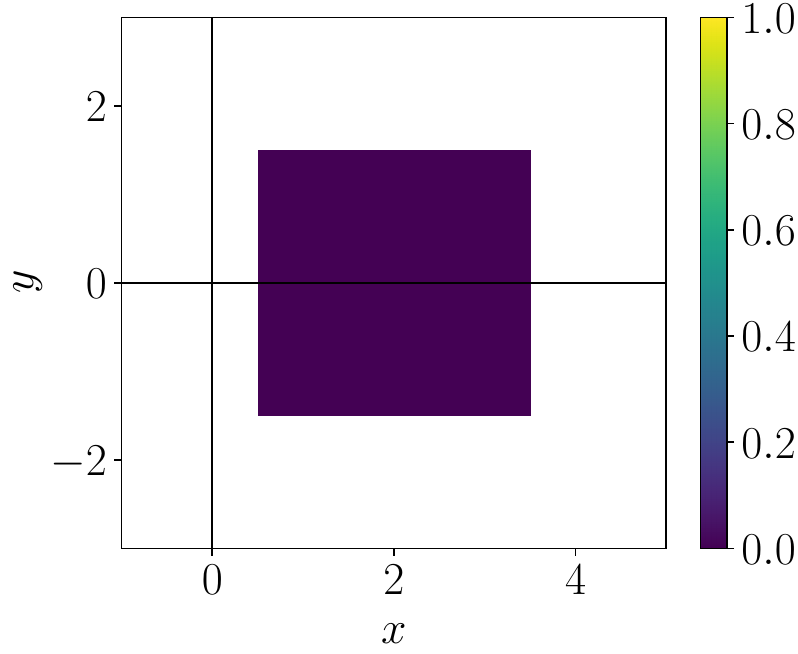}
        \subcaption{$I_1,~r=1.5$}
    \end{subfigure}
    \hfill
    \begin{subfigure}[t]{0.24\textwidth}
        \centering
        \includegraphics[width=\textwidth]{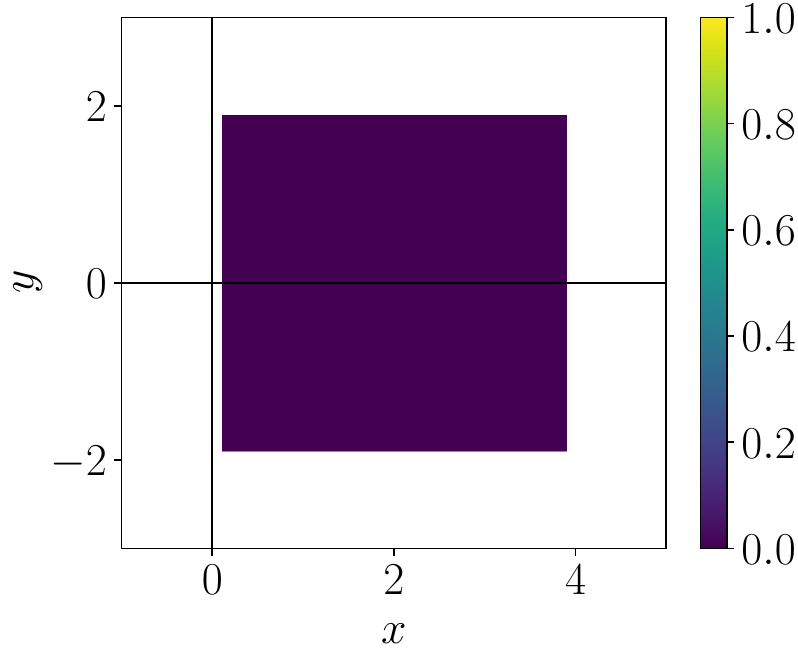}
        \subcaption{$I_1,~r=1.9$}        
    \end{subfigure}
    \begin{subfigure}[t]{0.24\textwidth}
        \centering
        \includegraphics[width=\textwidth]{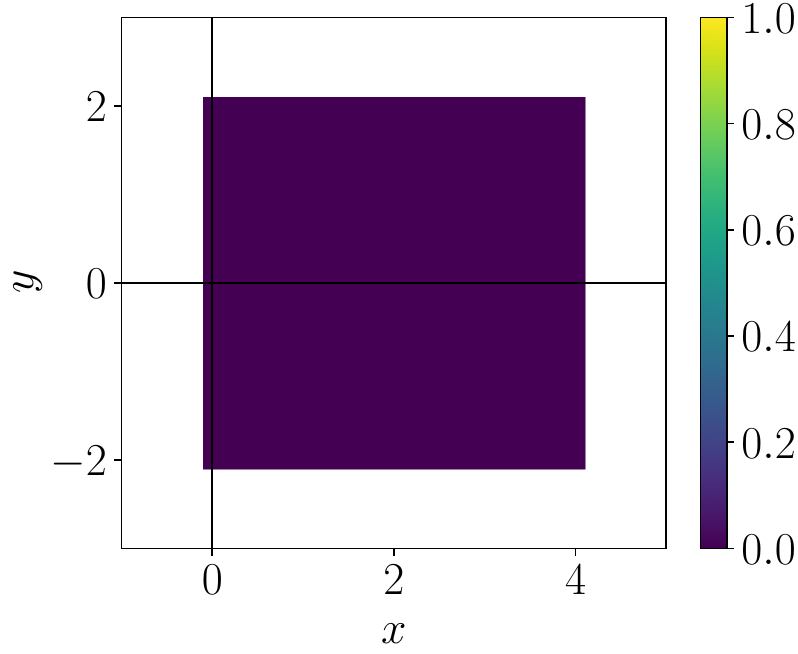}
        \subcaption{$I_1,~r=2.1$}        
    \end{subfigure}
    \hfill
    \begin{subfigure}[t]{0.24\textwidth}
        \centering
        \includegraphics[width=\textwidth]{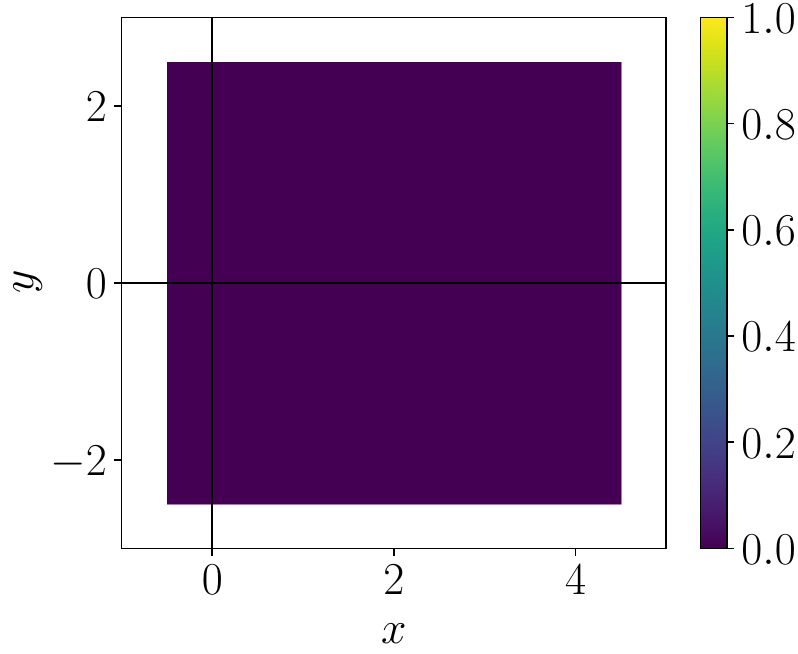}
        \subcaption{$I_1,~r=2.5$}        
    \end{subfigure}
    \hfill
        \begin{subfigure}[t]{0.24\textwidth}
        \centering
        \includegraphics[width=\textwidth]{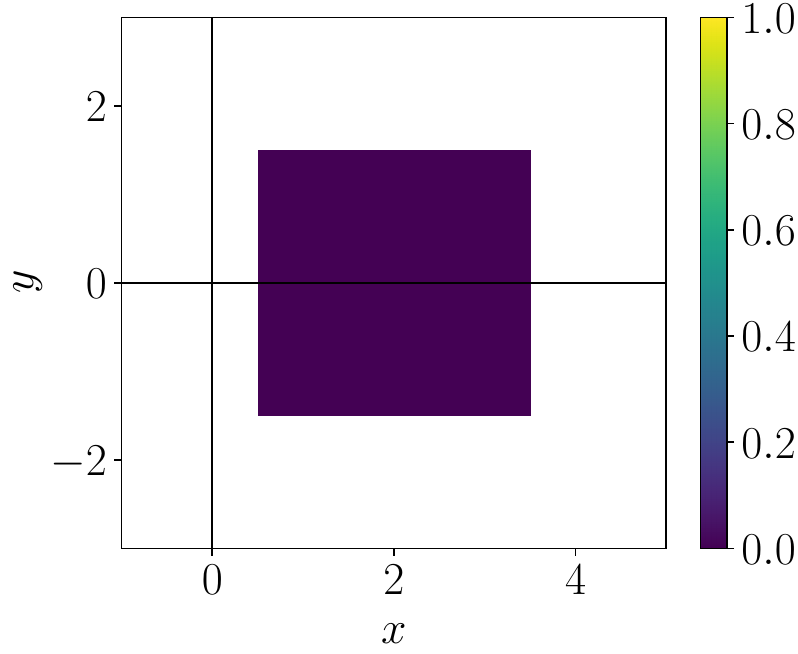}
        \subcaption{$I_2,~r=1.5$}
    \end{subfigure}
    \hfill
    \begin{subfigure}[t]{0.24\textwidth}
        \centering
        \includegraphics[width=\textwidth]{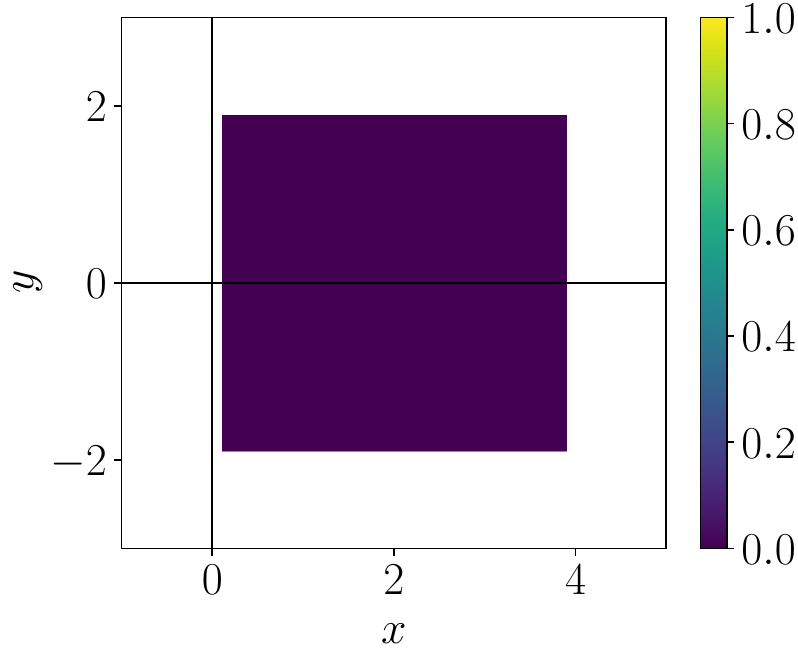}
        \subcaption{$I_2,~r=1.9$}        
    \end{subfigure}
    \begin{subfigure}[t]{0.24\textwidth}
        \centering
        \includegraphics[width=\textwidth]{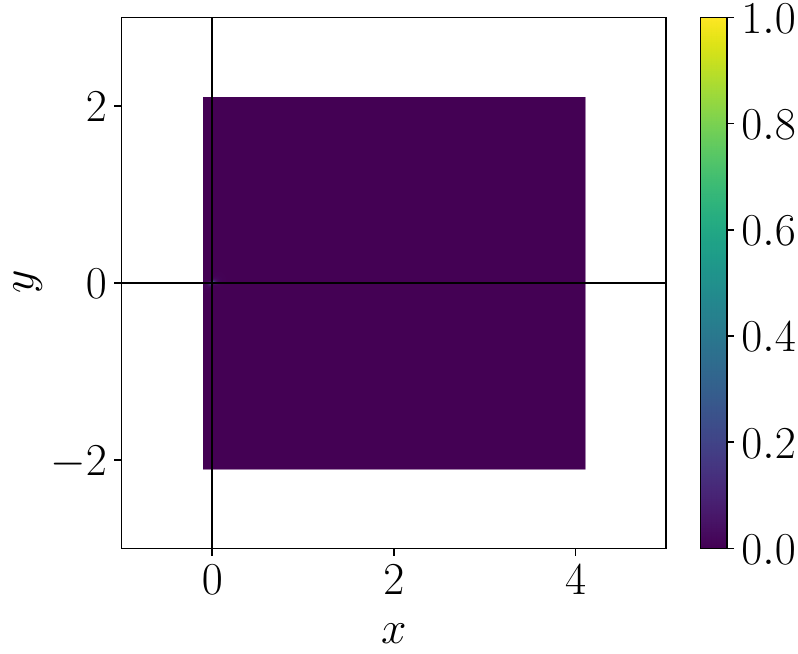}
        \subcaption{$I_2,~r=2.1$}        
    \end{subfigure}
    \hfill
    \begin{subfigure}[t]{0.24\textwidth}
        \centering
        \includegraphics[width=\textwidth]{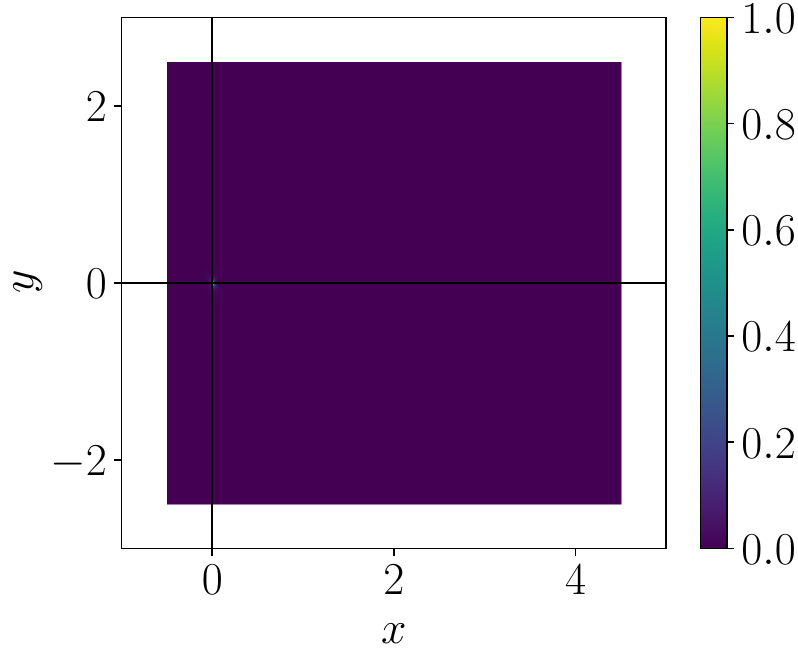}
        \subcaption{$I_2,~r=2.5$}        
    \end{subfigure}
    \caption{Pointwise losses for the 3D degenerate center
(Example~\ref{ex:center_3d}) on the slice $z=0$ of
$U_r^{(3)}=[2-r,2+r]\times[-r,r]^2$, with
$r=1.5,1.9,2.1,2.5$ from left to right. Panels (a)--(d) show the
conservation loss, equivalently the deflated loss, for $I_1$, while
panels (e)--(h) show the deflated loss for $I_2$. Unlike the 3D
degenerate node, both conservation laws persist after the box
intersects the equilibrium line.
}
    \label{fig:3d_center}
\end{figure}

Together, the node and center examples distinguish two different
effects of an equilibrium set. For the 2D node, the angular
conservation law cannot be continued across the origin, whereas the
energy of the 2D center remains globally defined despite having a
critical point there.
The same contrast persists in three dimensions, where
the global law $z$ survives in both degenerate systems,  yet the second
law breaks down for the degenerate node and persists for the
degenerate center. Neural deflation therefore detects a genuine loss
of a conservation law while explicitly distinguishing it from the vanishing of
the gradient of an otherwise valid invariant.

\subsection{H\'enon--Heiles system}
\label{subsec:henon_heiles}

We next consider the H\'enon--Heiles system~\cite{henon1964applicability,fordy1991henon}, a four-dimensional (two-degree-of-freedom) 
Hamiltonian system. We express its Hamiltonian in the following normalized form
\begin{align}
\label{eq:henon_heiles_hamiltonian}
    H_\epsilon(x,y,p_x,p_y)
    =
    \frac{1}{2}\left(p_x^2+p_y^2\right)
    +
    \frac{1}{2}\left(x^2+y^2\right)
    +
    x^2y-\epsilon y^3,
\end{align}
where $\epsilon$ is a free parameter. Throughout this subsection, we
consider the parameter range $\epsilon \ge -1/3$. The corresponding Hamiltonian equations  of motion are
\begin{align}
\label{eq:henon_heiles_system}
    \left\{
    \begin{aligned}
        \dot{x}   &= p_x,\\
        \dot{y}   &= p_y,\\
        \dot{p}_x &= -x-2xy,\\
        \dot{p}_y &= -y-x^2+3\epsilon y^2.
    \end{aligned}
    \right.
\end{align}

The Hamiltonian $H_\epsilon$ is a global conservation law for every
$\epsilon$. Within the parameter range considered here, the system is known to be 
Liouville integrable  only when $\epsilon=-1/3$. In this case, it
admits two functionally independent global conservation laws in
involution, one of which is $H_{-1/3}$. For
$\epsilon > -1/3$, the system is non-integrable, and the
Hamiltonian is the only global conservation law.

We study the system using two families of expanding domains. The first
consists of rectangular boxes centered at the regular point
$(1,1,0,0)$, allowing us to recover the three local conservation laws
predicted by the flow-box theorem and track their subsequent decrease. 
The
second consists of energy sublevel sets growing from the equilibrium
at the origin, allowing us to examine how the persistence of a second
conservation law depends on the distance of $\epsilon$ from the
integrable value $\epsilon=-1/3$.

\subsubsection{Rectangular boxes growing from a regular point}
\label{subsubsec:hh_rectangular}

\begin{figure}
    \centering
    \begin{subfigure}[t]{0.43\textwidth}
        \centering
        \includegraphics[width=\textwidth]{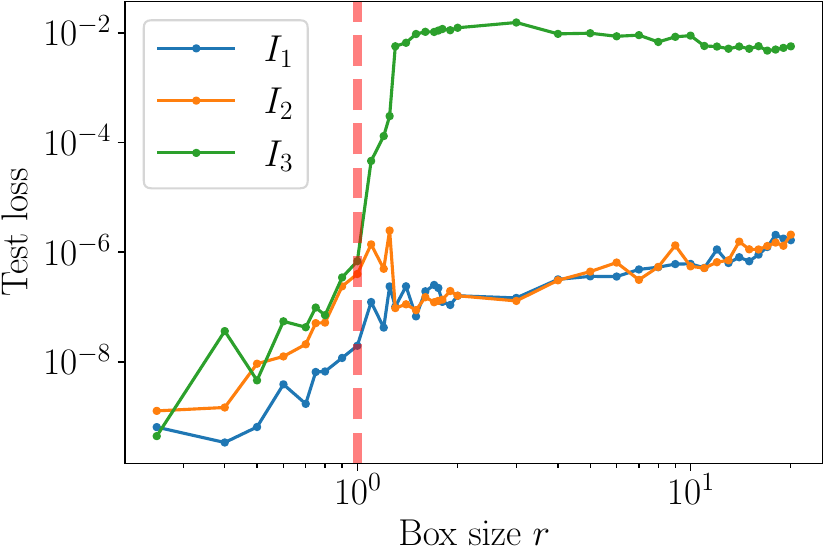}
        \subcaption{$\epsilon=-1/3$ (integrable)}
    \end{subfigure}
    \hfill
    \begin{subfigure}[t]{0.43\textwidth}
        \centering
        \includegraphics[width=\textwidth]{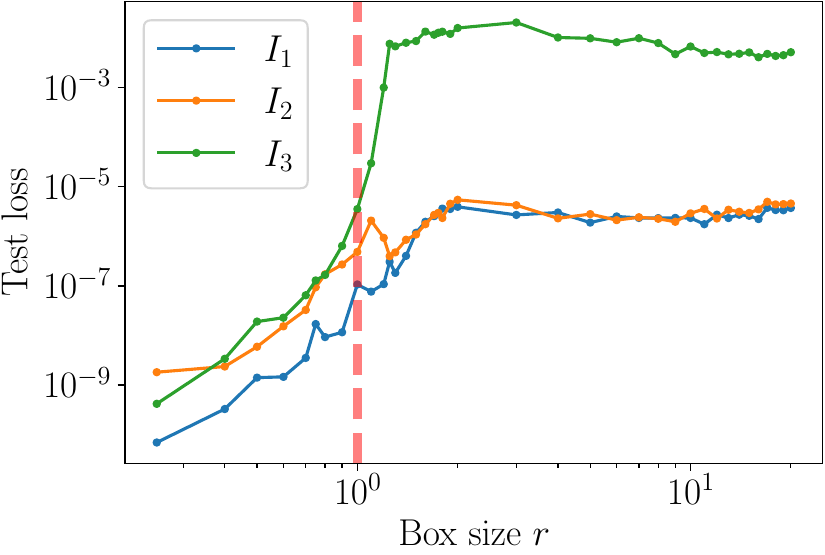}
        \subcaption{$\epsilon=-0.3$}
    \end{subfigure}    
    \hfill
    \begin{subfigure}[t]{0.43\textwidth}
        \centering
        \includegraphics[width=\textwidth]{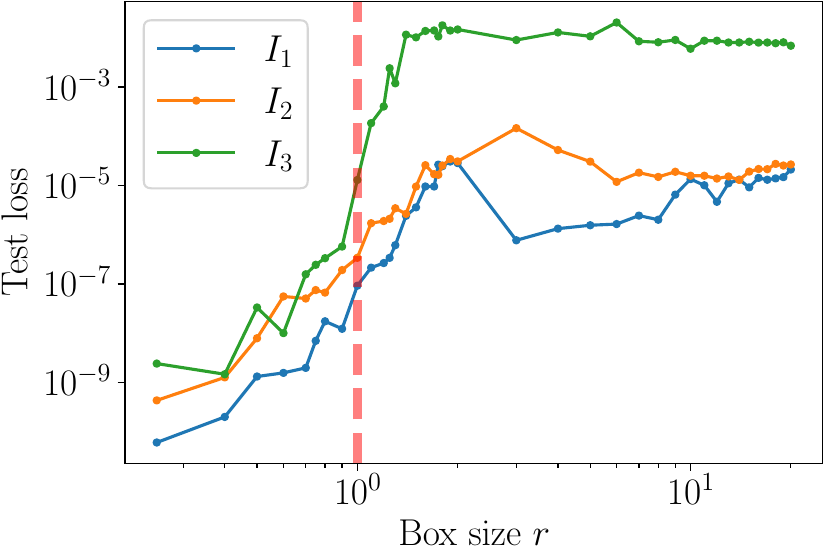}
        \subcaption{$\epsilon=-0.2$}
    \end{subfigure}
    \hfill
    \begin{subfigure}[t]{0.43\textwidth}
        \centering
        \includegraphics[width=\textwidth]{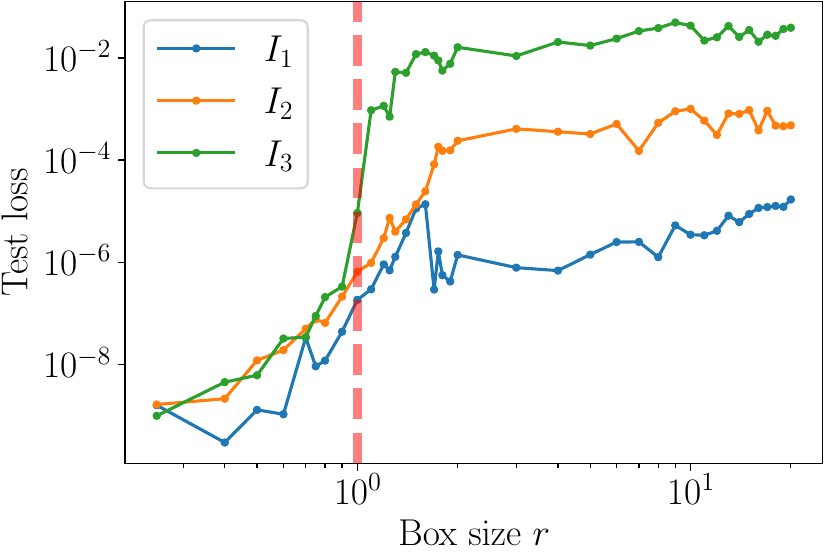}
        \subcaption{$\epsilon=-0.1$}        
    \end{subfigure}
        \hfill
    \begin{subfigure}[t]{0.43\textwidth}
        \centering
        \includegraphics[width=\textwidth]{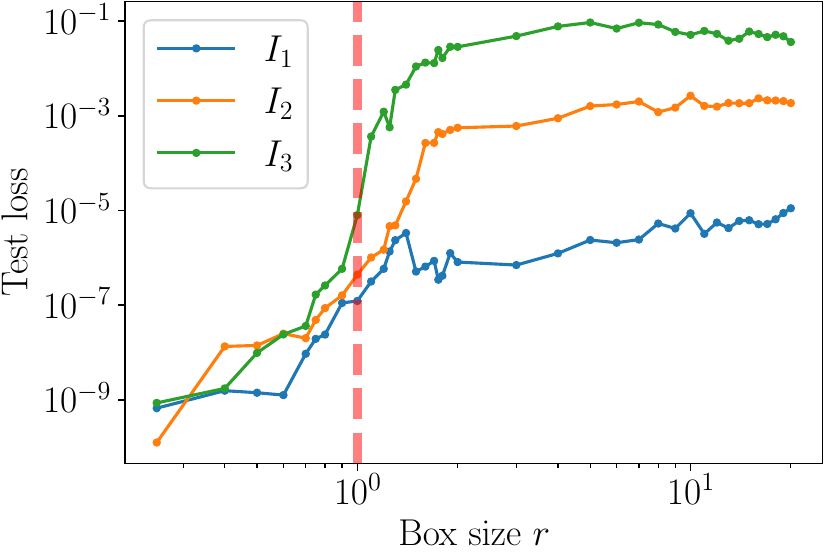}
        \subcaption{$\epsilon=0$}
    \end{subfigure}
    \hfill
    \begin{subfigure}[t]{0.43\textwidth}
        \centering
        \includegraphics[width=\textwidth]{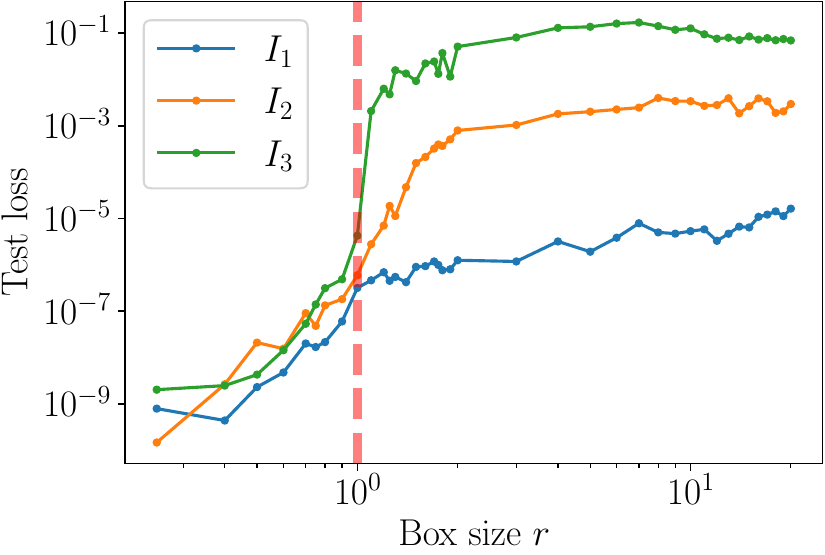}
        \subcaption{$\epsilon=0.1$}        
    \end{subfigure}        
    \caption{Domain-averaged deflated losses of the three learned candidates for
the H\'enon--Heiles system on the expanding boxes
$U_r^{\mathrm{rec}}=(1,1,0,0)+[-r,r]^4$. Panels (a)--(f) correspond
to $\epsilon=-1/3,-0.3,-0.2,-0.1,0,$ and $0.1$, respectively. Three
conservation laws are recovered before the box reaches the equilibrium
at $r=1$, after which the third is lost. The two remaining laws persist
for the integrable case $\epsilon=-1/3$. For $\epsilon=-0.3$ and
$\epsilon=-0.2$, no clear loss of the second candidate is observed
over the tested range of $r$, while a pronounced separation between the first
and second losses emerges from $\epsilon=-0.1$ onward.}
    \label{fig:hh_rectangular_losses}
\end{figure}

We begin from the regular point $(1,1,0,0)$ and consider the expanding
boxes
\begin{align}
\label{eq:hh_rectangular_domains}
    U_r^{\mathrm{rec}}
    =
    (1,1,0,0)+[-r,r]^4.
\end{align}
For sufficiently small $r$, the flow-box theorem predicts three
functionally independent local conservation laws. We therefore train
three candidates sequentially and examine their domain-averaged
deflated losses as $r$ increases.

Figure~\ref{fig:hh_rectangular_losses} shows that all three candidates
have small losses of comparable magnitude before the box reaches the origin. At the critical
value $r=1$, the box first intersects the equilibrium
$(0,0,0,0)$, where the vector field vanishes and the flow-box
description ceases to apply. Correspondingly, the loss of the third candidate increases, indicating that the number of recovered conservation laws decreases from three to fewer than three. This first transition occurs in both the
integrable and non-integrable cases and is essentially independent of
$\epsilon$.

Figure~\ref{fig:hh_flowbox_breakdown} examines the first transition
in the number of learned conservation laws for the integrable case
$\epsilon=-1/3$ and the non-integrable case $\epsilon=-0.2$. We plot
the pointwise deflated loss of the third candidate $I_3$ on the
$p_x$--$p_y$ slice $x=y=1-r$ of the boundary of
$U_r^{\mathrm{rec}}$. At $r=0.9$, the boundary has not yet reached
the origin, and the loss remains small for both values of $\epsilon$.
At $r=1$, the boundary first contains the equilibrium
$(0,0,0,0)$, and the loss increases sharply. Thus, independently of
integrability, the flow-box regime breaks down when the expanding box
reaches the equilibrium, reducing the number of recovered conservation
laws from three to two.

\begin{figure}
    \centering
    \begin{subfigure}[t]{0.45\textwidth}
        \centering
        \includegraphics[width=\textwidth]{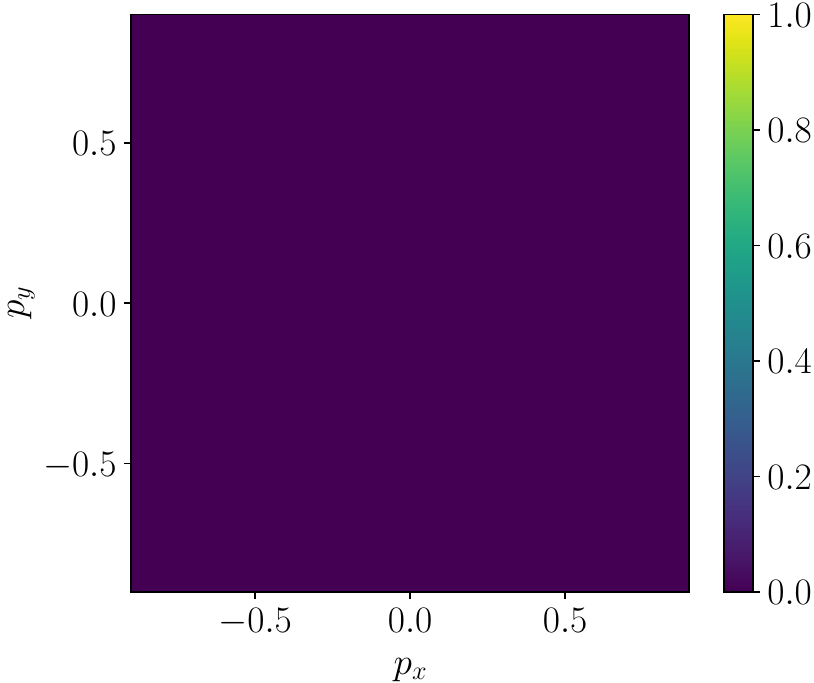}   \subcaption{$\epsilon = -1/3$, $r=0.9$}
    \end{subfigure}
    \begin{subfigure}[t]{0.45\textwidth}
        \centering
        \includegraphics[width=\textwidth]{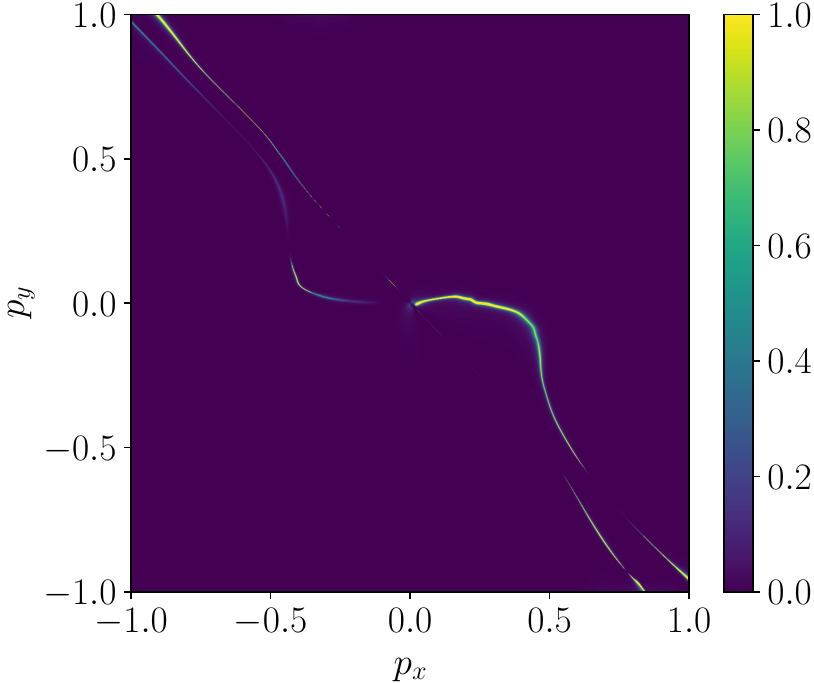}   \subcaption{$\epsilon = -1/3$, $r=1.0$}           
    \end{subfigure}    
    \begin{subfigure}[t]{0.45\textwidth}
        \centering
        \includegraphics[width=\textwidth]{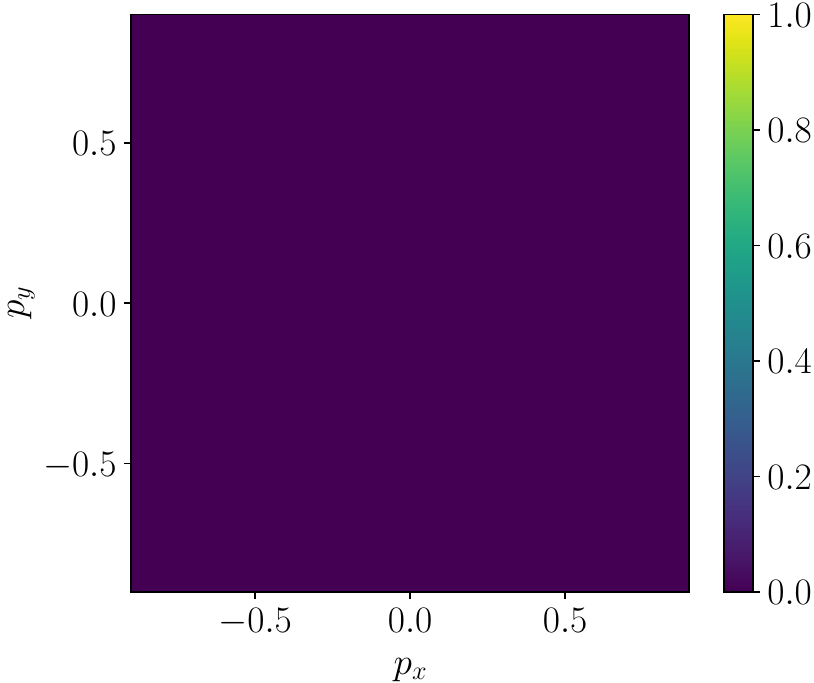}   \subcaption{$\epsilon = -0.2$, $r=0.9$}
    \end{subfigure}
    \begin{subfigure}[t]{0.45\textwidth}
        \centering
        \includegraphics[width=\textwidth]{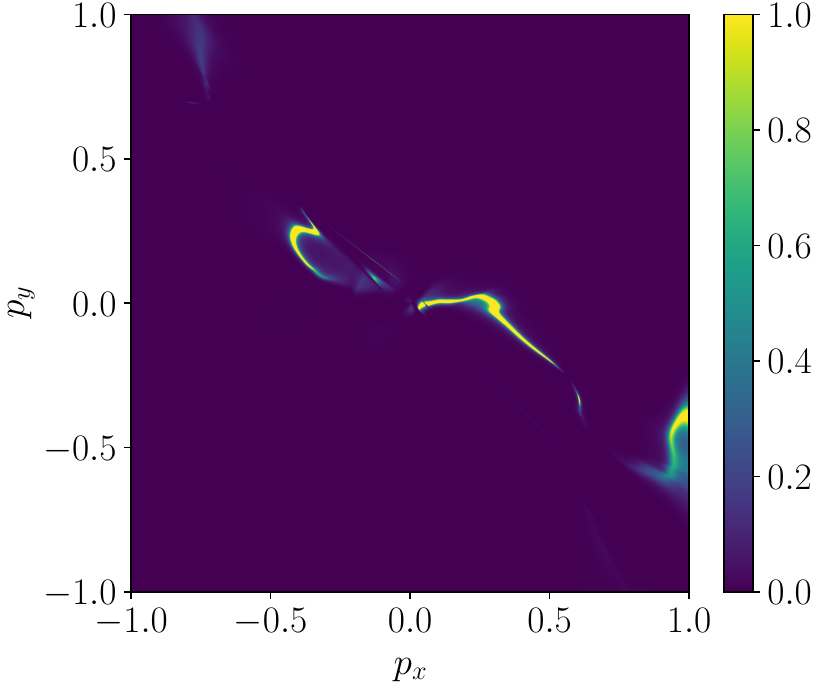}   \subcaption{$\epsilon = -0.2$, $r=1.0$}           
    \end{subfigure}
\caption{Pointwise deflated loss of the third learned conservation law $I_3$
for the H\'enon--Heiles system on the $p_x$--$p_y$ slice
$x=y=1-r$ of the boundary of $U_r^{\mathrm{rec}}$. Panels (a) and
(b) correspond to $\epsilon=-1/3$ with $r=0.9$ and $r=1.0$,
respectively, while panels (c) and (d) show the corresponding results
for $\epsilon=-0.2$. For both parameter values, the loss is small at
$r=0.9$ but increases sharply at $r=1$, when the boundary first
reaches the equilibrium at the origin, marking the breakdown of the third
local conservation law.}
\label{fig:hh_flowbox_breakdown}
\end{figure}
Beyond the initial transition at $r=1$, the behavior depends strongly on $\epsilon$. In the integrable case $\epsilon=-1/3$, the two remaining losses stay small as the box continues to expand, consistent with the existence of two independent global conservation laws in involution; see Figure~\ref{fig:hh_rectangular_losses}(a).

For the nearby non-integrable cases $\epsilon=-0.3$ and $\epsilon=-0.2$, the second candidate also maintains a small loss throughout the tested range of box sizes; see Figure~\ref{fig:hh_rectangular_losses}(b)--(c). Thus, although these systems are not globally integrable, the present experiment does not numerically resolve the eventual breakdown of the second conservation law. 
We do note here in passing that for some intermediate box sizes in 
case (c) a nontrivial difference between the first and second conservation
law loss appears, yet we cannot use this as a systematic enough indication
of the loss of integrability (given that this disparity decreases for
even larger boxes).
A clear qualitative change first appears at $\epsilon=-0.1$; see Figure~\ref{fig:hh_rectangular_losses}(d). As the box expands, the loss of the second candidate increases substantially, whereas the first candidate, corresponding to the Hamiltonian, remains accurate. This separation becomes even more pronounced for $\epsilon=0$ and $\epsilon=0.1$; see Figure~\ref{fig:hh_rectangular_losses}(e)--(f). Accordingly,
we can definitively infer the non-integrability of the 
cases in panels (d)--(f), although once again this is not a 
sharp enough criterion to operate as a precise detector of
integrability, to the best of our understanding.

These results suggest that the second conservation law persists over increasingly large domains as $\epsilon$ approaches the integrable value $-1/3$. For $\epsilon=-0.3$ and $\epsilon=-0.2$, its breakdown, if detectable by the present method, may occur only beyond the range of box sizes considered here.

\subsubsection{Energy sublevel sets growing from the origin}
\label{subsubsec:hh_energy}

We next expand the training domain outward from the equilibrium at the
origin, using energy sublevel sets rather than rectangular boxes.
Since the full energy sublevel set need not be bounded, we restrict
attention to its bounded connected component containing the origin.
Specifically, for $\epsilon\in[-1/3,0]$ and an energy level $h>0$,
let $U_h^{\mathrm{en}}(\epsilon)$ denote this component of
\begin{align}
\label{eq:hh_energy_domains}
    \left\{
        (x,y,p_x,p_y):
        H_\epsilon(x,y,p_x,p_y)\leq h
    \right\}.
\end{align}
The first hyperbolic equilibria are
\begin{align}
\label{eq:hh_first_hyperbolic_equilibria}
    \left(
        \pm\sqrt{\frac{3}{4}\epsilon+\frac{1}{2}},
        -\frac{1}{2},
        0,
        0
    \right),
\end{align}
with energy
\begin{align}
\label{eq:hh_escape_energy}
    h_{\mathrm{esc}}(\epsilon)
    =
    \frac{1+\epsilon}{8}.
\end{align}
This defines the escape threshold, below which the connected
component containing the origin remains bounded.

To compare different values of $\epsilon$ on a common scale,
we parameterize the energy level by
\begin{align}
\label{eq:hh_normalized_energy}
    h=r(1+\epsilon),
    \qquad
    0\leq r\leq \frac{1}{8}.
\end{align}
For each $\epsilon$, choosing $r$ specifies the energy level
$h=r(1+\epsilon)$, whose level set defines the boundary of
the corresponding bounded training domain
$U_{r(1+\epsilon)}^{\mathrm{en}}(\epsilon)$.
As $r$ increases, this domain expands toward the escape
threshold, which corresponds to $r=1/8$ independently of
$\epsilon$. Thus, the horizontal axis labeled $r$ in
Figure~\ref{fig:HH_energy_sublevelset} represents normalized
energy rather than the energy itself. Samples are drawn
uniformly from $U_{r(1+\epsilon)}^{\mathrm{en}}(\epsilon)$.

Because every such domain contains the equilibrium at the origin, the
flow-box theorem does not predict three independent conservation laws.
Indeed, the third candidate has a large deflated loss from the outset.
The main question is therefore how long a second low-loss candidate
persists as the energy domain expands.

Figure~\ref{fig:HH_energy_sublevelset} shows the domain-averaged
deflated losses for nine values of $\epsilon$. In the integrable case
$\epsilon=-1/3$, the first two candidates retain small losses
throughout the tested range, consistent with the existence of two
independent global conservation laws in involution. The same numerical
behavior remains visible for the nearby non-integrable cases
$\epsilon=-0.3$ and $\epsilon=-0.27$: although no second global
conservation law exists, its eventual breakdown is not resolved over
the tested energy range. I.e., similarly to the case of
the rectangular boxes of Figure~\ref{fig:hh_rectangular_losses},
the energy sublevels are unable to sharply detect the loss
of integrability near the integrable case of $\epsilon=-1/3$.

A clear loss of the second candidate first becomes visible at
$\epsilon=-0.25$, where its deflated loss begins to increase at larger
values of $r$. As $\epsilon$ moves farther from the integrable value
$-1/3$, this increase occurs at progressively smaller normalized
energy levels. This trend is evident for
$\epsilon=-0.24,-0.21,$ and $-0.2$. By $\epsilon=-0.1$, the second
candidate already has a large loss at the smallest tested domains, so
only the Hamiltonian is recovered from the outset. The same behavior
persists at $\epsilon=0$.

These results suggest that proximity to the integrable limit controls the
numerical persistence of the second candidate. Near $\epsilon=-1/3$,
it remains accurate over a large range of energy sublevel sets; farther
away, it breaks down on increasingly smaller domains. Since this loss
can occur well before $r=1/8$, it is not simply caused by the energy
contour reaching the first hyperbolic equilibrium. 
While the
latter could have been an appealing threshold for the identification
of loss of integrability and the loss of the associated second conservation
law, such an interpretable scenario has not been verified 
by our numerical computations.

\begin{figure}
    \centering
    \begin{subfigure}[t]{0.32\textwidth}
        \centering
        \includegraphics[width=\textwidth]{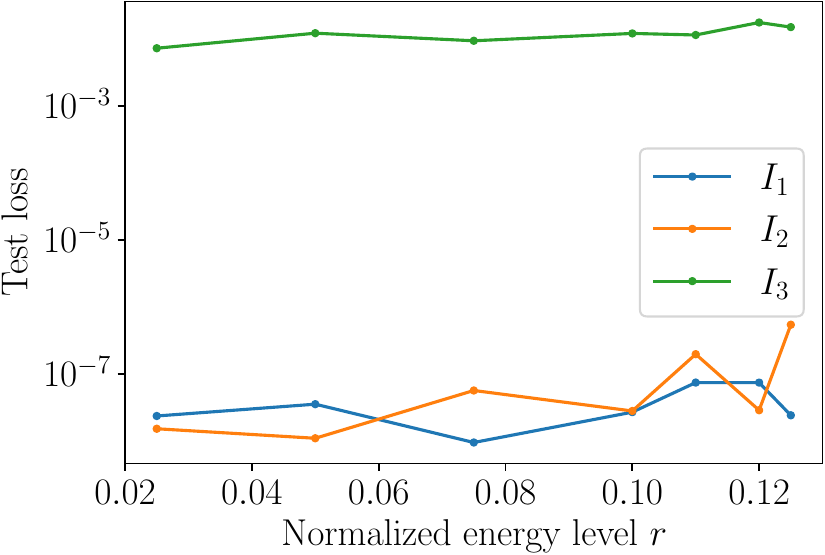}
        \subcaption{$\epsilon=-1/3$ (integrable)}
    \end{subfigure} 
    \begin{subfigure}[t]{0.32\textwidth}
        \centering
        \includegraphics[width=\textwidth]{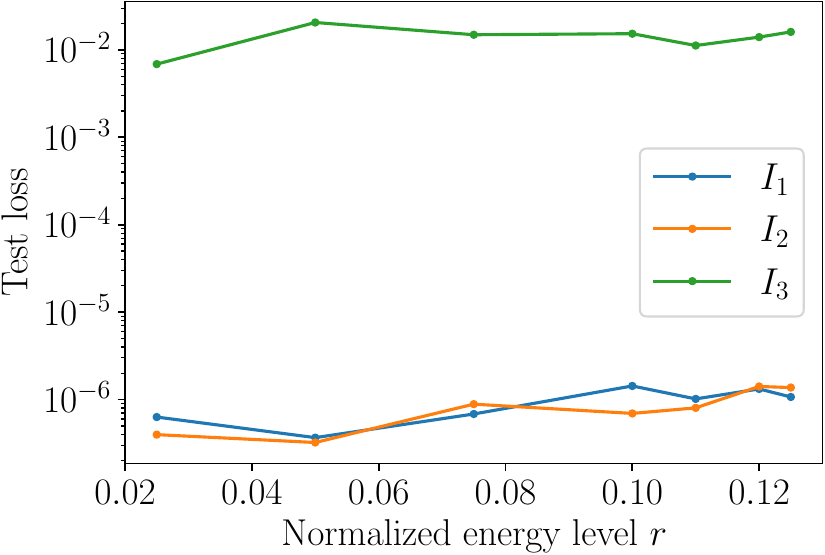}
        \subcaption{$\epsilon=-0.3$}
    \end{subfigure}
    \begin{subfigure}[t]{0.32\textwidth}
        \centering
        \includegraphics[width=\textwidth]{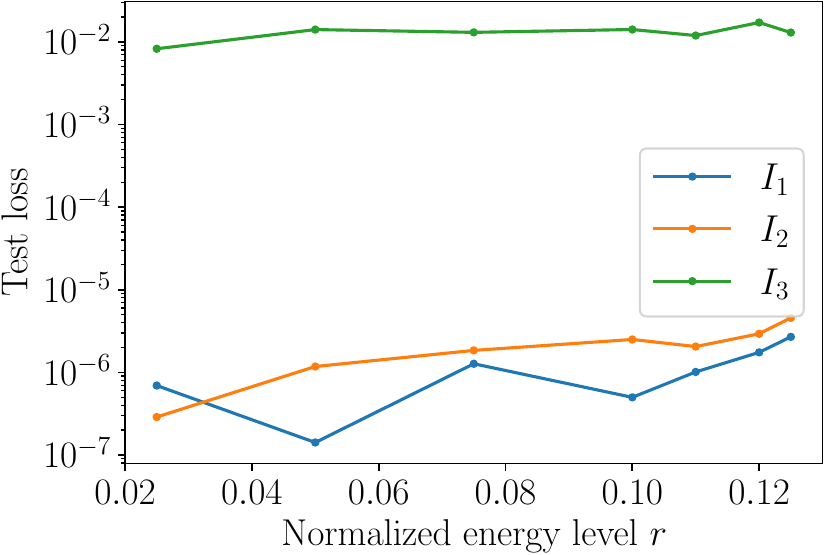}
        \subcaption{$\epsilon=-0.27$}
    \end{subfigure}    
    \begin{subfigure}[t]{0.32\textwidth}
        \centering
        \includegraphics[width=\textwidth]{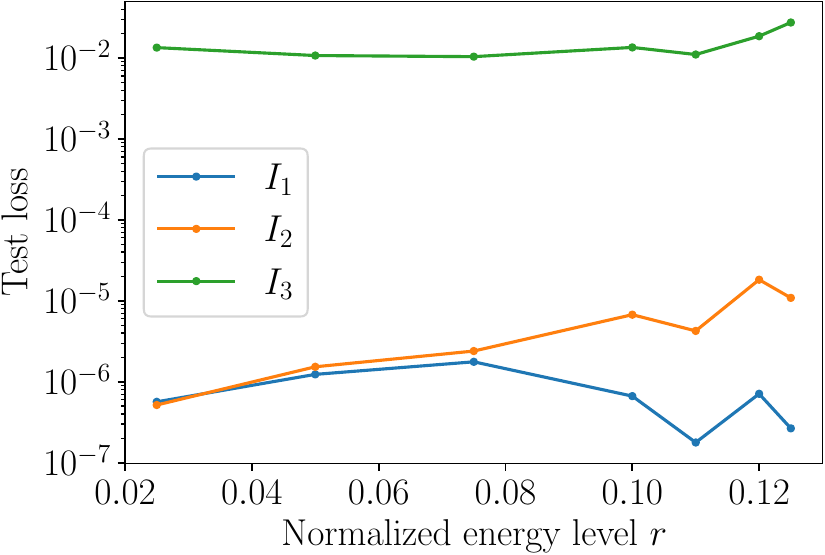}
        \subcaption{$\epsilon=-0.25$}
    \end{subfigure}
    \begin{subfigure}[t]{0.32\textwidth}
        \centering
        \includegraphics[width=\textwidth]{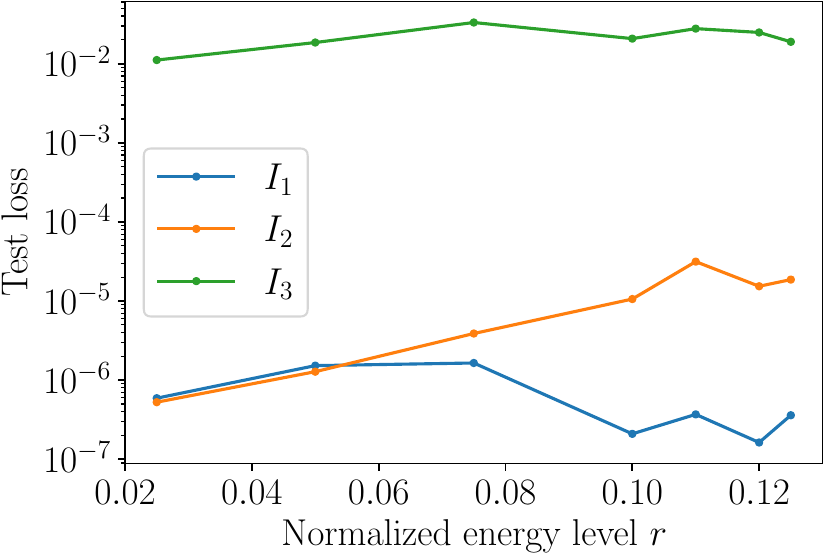}
        \subcaption{$\epsilon=-0.24$}
    \end{subfigure}     
    \begin{subfigure}[t]{0.32\textwidth}
        \centering
        \includegraphics[width=\textwidth]{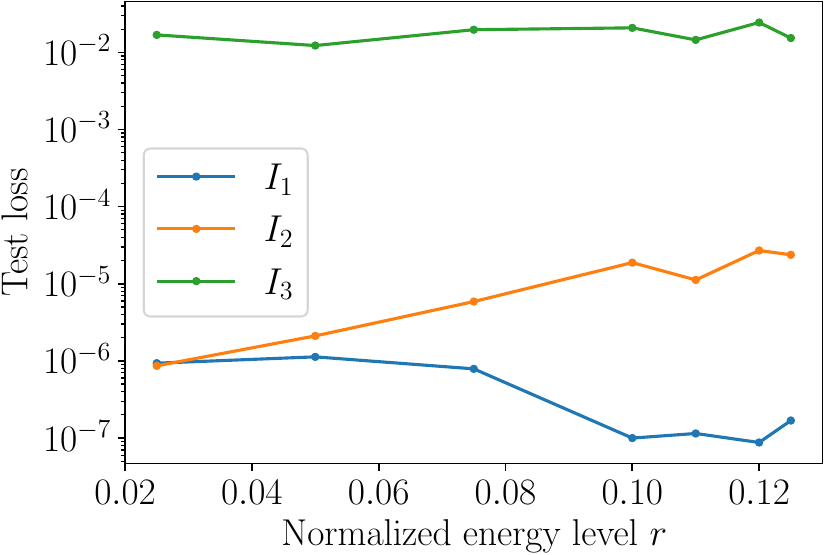}
        \subcaption{$\epsilon=-0.21$}
    \end{subfigure}   
    \begin{subfigure}[t]{0.32\textwidth}
        \centering
        \includegraphics[width=\textwidth]{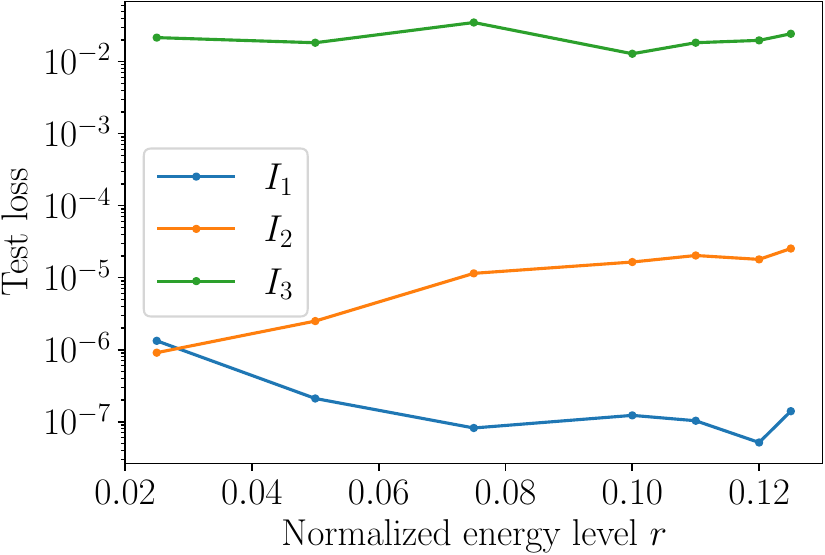}
        \subcaption{$\epsilon=-0.2$}
    \end{subfigure}      
    \begin{subfigure}[t]{0.32\textwidth}
        \centering
        \includegraphics[width=\textwidth]{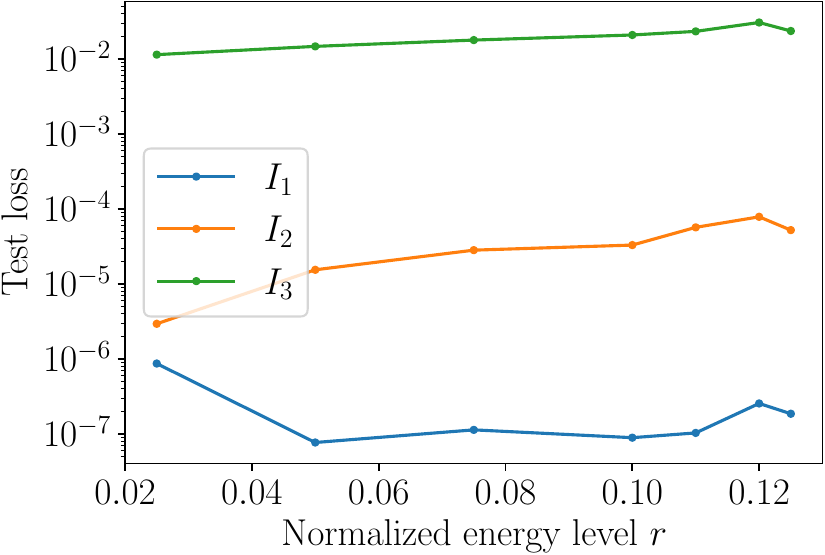}
        \subcaption{$\epsilon=-0.1$}
    \end{subfigure}
    \begin{subfigure}[t]{0.32\textwidth}
        \centering
        \includegraphics[width=\textwidth]{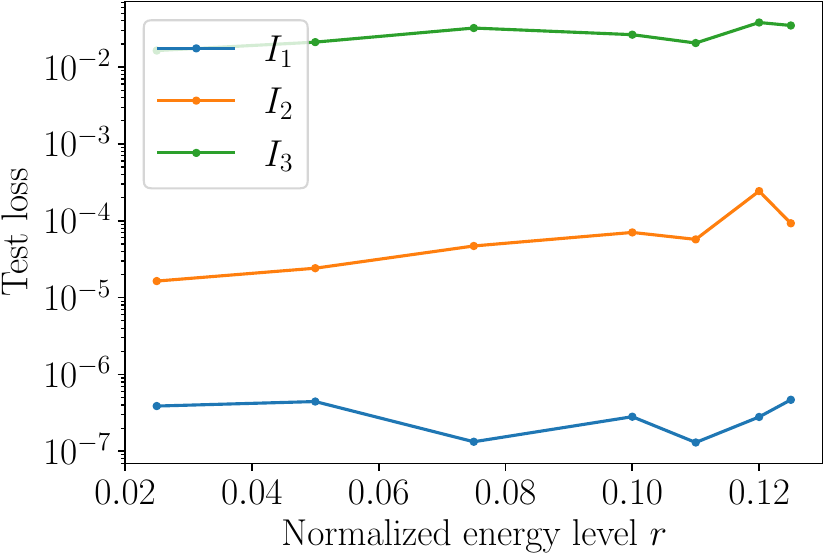}
        \subcaption{$\epsilon=0$}
    \end{subfigure}
    \caption{
    Domain-averaged deflated losses of the three learned candidates for
    the H\'enon--Heiles system on the energy sublevel sets
    $U_{r(1+\epsilon)}^{\mathrm{en}}(\epsilon)$ growing from the
    origin. Panels (a)--(i) correspond to
    $\epsilon=-1/3,-0.3,-0.27,-0.25,-0.24,-0.21,-0.2,-0.1,$ and
    $0$, respectively. The horizontal axis is the normalized energy
    parameter $r$, with energy level $h=r(1+\epsilon)$. Two candidates
    persist throughout the tested range in the integrable case. The
    breakdown of the second candidate first becomes visible at
    $\epsilon=-0.25$ and occurs at progressively smaller values of
    $r$ as $\epsilon$ moves away from $-1/3$; by $\epsilon=-0.1$, it
    is lost from the smallest tested domain.
    }
    \label{fig:HH_energy_sublevelset}
\end{figure}

The two domain-expansion strategies reveal complementary aspects of
the H\'enon--Heiles dynamics. Rectangular boxes growing from a regular
point initially recover the three local conservation laws predicted by
the flow-box theorem and then track their successive failure. Energy
sublevel sets growing from the equilibrium contain no three-law
flow-box regime, but more directly reveal how the persistence of a
second candidate depends on proximity to the integrable limit.

\subsubsection{Conservation loss near the integrable limit}
\label{subsubsec:hh_predictor}
The numerical results of the previous subsection show that the second candidate conservation law persists over larger domains as $\epsilon$ approaches the integrable limit. For the H\'enon--Heiles system, we find that this dependence can be compared directly with the residual of an explicit invariant of the integrable system. 

To begin,  write
\begin{equation}
\label{eq:pred_split}
H_\epsilon=H_\star+\mu P,\qquad
\mu=\epsilon+\tfrac13,\qquad P=-y^3,
\end{equation}
where $H_\star=H_{-1/3}$ is the integrable Hamiltonian. For convenience in the calculations, we use the standard canonical coordinates
$q_\pm=y\pm x$, $p_\pm=\tfrac12(p_y\pm p_x)$ so that
\begin{equation}
\label{eq:pred_sep}
H_\star=A_++A_-,\qquad A_\pm=p_\pm^2+W(q_\pm),
\end{equation}
where $W(q)=\tfrac14q^2+\tfrac16q^3.$ It is straightforward to show that the difference
\[
D=A_+-A_-
\]
provides a second integral that is functionally independent of $H_\star$ almost everywhere.

\begin{figure}[ht]
    \centering
    \includegraphics[width=.7\columnwidth]
    {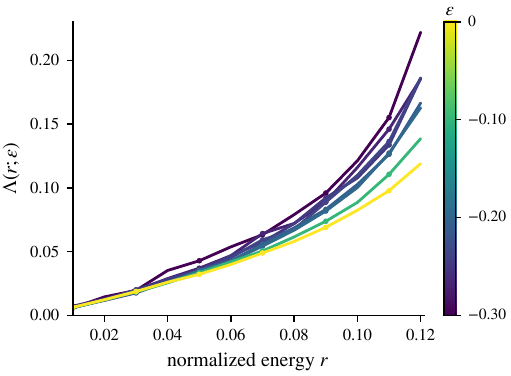}
    \caption{
    Dependence of $\Lambda(r;\epsilon)$ on the normalized energy $r$ for
the H\'enon--Heiles system. At fixed $r$, the variation with $\epsilon$
is modest, while $\Lambda$ increases substantially as the energy domain
expands.
    }
    \label{fig:hh_near_integrable_scaling}
\end{figure}

Now, to connect to neural deflation, use $\{D,H_\star\}=0$ and $\partial D/\partial p_y=p_x$ so that
\begin{equation}
\label{eq:pred_identity}
\mathbf f_\epsilon\cdot\nabla D
=\{D,H_\epsilon\}=\mu\{D,P\}=3\mu y^2p_x.
\end{equation}
Thus, the normalized conservation loss in
Equation~\eqref{eq:pointwise_conservation_loss} can be expressed as
\begin{equation}
\label{eq:pred_lcons}
\ell_{\mathrm{cons}}[D;z]
=\frac{9\mu^2y^4p_x^2}
{\|\mathbf f_\epsilon(z)\|_2^2\|\nabla D(z)\|_2^2}.
\end{equation}
Here we take norms and gradients in the original coordinates $z=(x,y,p_x,p_y)$. For the deflation exponent $\alpha=1$ used in the previous experiments, substituting $I_1=H_\epsilon$ gives
\begin{equation}
\label{eq:pred_deflated}
\ell_{\mathrm{defl},2}[D;z]
=\frac{9\mu^2y^4p_x^2}{Q_\epsilon(z)},
\end{equation}
where
\begin{equation}
\label{eq:pred_gram}
Q_\epsilon
=\|\nabla H_\epsilon\|_2^2\|\nabla D\|_2^2
-(\nabla H_\epsilon\cdot\nabla D)^2.
\end{equation}
Equation~\eqref{eq:pred_deflated} follows from
$\|\mathbf f_\epsilon\|_2=\|\nabla H_\epsilon\|_2$ and
Equation~\eqref{eq:pointwise_deflated_loss}, and holds wherever $Q_\epsilon>0$.

For $U=U^{\mathrm{en}}_{r(1+\epsilon)}(\epsilon)$, let
$\langle\cdot\rangle_U$ denote the uniform domain average and define
\begin{equation}
\label{eq:pred_Lambda}
\Lambda(r;\epsilon)
=
\left\langle
\frac{9y^4p_x^2}{Q_\epsilon(z)}
\right\rangle_U,
\qquad
\mathcal L_D(r;\epsilon)
=
\mu^2\Lambda(r;\epsilon).
\end{equation}
Here $\mathcal L_D$ is the domain-averaged deflated loss of the explicit second integral $D$. To compute the averages, we use Monte Carlo sampling on the energy sublevel sets
$U_{r(1+\epsilon)}^{\mathrm{en}}(\epsilon)$ with
$\epsilon\in\{-0.30,-0.27,-0.25,-0.24,-0.21,-0.20,-0.10,0\}$ and
$r=0.01,0.02,\ldots,0.12$. For each pair $(\epsilon,r)$, we drew
$10^5$ samples uniformly from the bounded component containing the
origin and evaluated the integrand in Equation~\eqref{eq:pred_Lambda}.
The reported value of $\Lambda(r;\epsilon)$ is the mean of three
independent Monte Carlo realizations.

Figure~\ref{fig:hh_near_integrable_scaling} displays one such set of computations across the pairs $(\epsilon,r)$, revealing that
$\Lambda(r;\epsilon)$ depends much more strongly on the domain size
than on $\epsilon$. Through $r=0.1$, its variation across the tested
values of $\epsilon$ is less than a factor of $1.5$, while it increases
by about an order of magnitude between $r=0.02$ and $r=0.12$. Since
$\mathcal L_D=\mu^2\Lambda$ with
$\mu=\epsilon+1/3$, the growth of the loss with $r$ is governed mainly
by $\Lambda$, whereas its separation across $\epsilon$ is governed
mainly by $\mu^2$, which changes by a factor of $100$ between
$\epsilon=-0.30$ and $\epsilon=0$.

With $I_1=H_\epsilon$, $\mathcal L_D$ provides an explicit
benchmark for the second-candidate deflated loss, obtained
by evaluating the fixed observable $D$ inherited from the
integrable system. Over the tested parameter range, the
increase of this benchmark with domain size and distance
from the integrable limit qualitatively parallels the
behavior of the learned second candidate in
Figure~\ref{fig:HH_energy_sublevelset}. This qualitative
agreement is consistent with interpreting the learned
persistence as a manifestation of near-integrable structure.

We consider a complementary study in Appendix~\ref{app:quad}, where we use a trajectory-based measure of the near-integrability to track the drift of an invariant of the integrable system along perturbed trajectories. We obtain a finite-time conservation defect of order $O(\mu^2)$, and include this calculation as an exploratory comparison with the domain-averaged analysis above.

\subsection{Double pendulum}
\label{subsec:double_pendulum}

Lastly, we consider a planar double pendulum. Let $m_1$ and $m_2$
denote the two masses, $l_1$ and $l_2$ the lengths of the two
pendulum arms, and $g$ the gravitational acceleration. For simplicity we consider equal masses and
equal lengths and choose dimensionless units such that
\[
    m_1=m_2=l_1=l_2=g=1.
\]
The angular positions of the two arms are denoted by $\theta_1$ and
$\theta_2$, and their angular velocities by
$\omega_1=\dot{\theta}_1$ and $\omega_2=\dot{\theta}_2$. We also write
\[
    \Delta=\theta_1-\theta_2
\]
for their relative angle. In the state variables
$(\theta_1,\theta_2,\omega_1,\omega_2)$, the equations of motion are
\begin{align}
\label{eq:double_pendulum}
    \left\{
    \begin{aligned}
        \dot{\theta}_1
        &=
        \omega_1,\\
        \dot{\theta}_2
        &=
        \omega_2,\\
        \dot{\omega}_1
        &=
        \frac{
        -3\sin\theta_1-\sin(\theta_1-2\theta_2)
        -2\sin\Delta
        \left(\omega_2^2+\omega_1^2\cos\Delta\right)
        }{
        3-\cos(2\Delta)
        },\\
        \dot{\omega}_2
        &=
        \frac{
        2\sin\Delta
        \left(
        2\omega_1^2+2\cos\theta_1+\omega_2^2\cos\Delta
        \right)
        }{
        3-\cos(2\Delta)
        }.
    \end{aligned}
    \right.
\end{align}
The system is Hamiltonian and consequently has the global conservation law
\begin{align}
\label{eq:double_pendulum_energy}
    H_{\mathrm{dp}}
    =
    \omega_1^2
    +\frac{1}{2}\omega_2^2
    +\omega_1\omega_2\cos\Delta
    -2\cos\theta_1-\cos\theta_2.
\end{align}
The four equilibria
\[
    (0,0,0,0),\qquad
    (0,\pi,0,0),\qquad
    (\pi,0,0,0),\qquad
    (\pi,\pi,0,0)
\]
have energies $-3,-1,1,$ and $3$, respectively. The origin is the
stable equilibrium, while the first separatrix is associated with the
saddle energy $H_{\mathrm{dp}}=-1$. Unlike the integrable
H\'enon--Heiles case considered above, the double pendulum is 
non-integrable and possesses only one global conservation law, namely
the energy.

As before, we use two families of expanding training domains. We first
grow rectangular boxes from a regular point and track the successive
loss of the three local conservation laws predicted by the flow-box
theorem. We then grow energy sublevel sets directly from the stable
equilibrium and examine the gradual loss of a second numerically
recovered candidate.

\subsubsection{Rectangular boxes growing from a regular point}
\label{subsubsec:double_pendulum_rectangular}

We center the rectangular domains at $\bx_0=(0,0,0.2,0)$ and define
\begin{align}
\label{eq:double_pendulum_rectangular_domains}
    U_r^{\mathrm{rec}}
    =
    [-r,r]\times[-r,r]
    \times[0.2-r,0.2+r]\times[-r,r].
\end{align}
The first critical radius is
\begin{align}
\label{eq:double_pendulum_r1}
    r_1=0.2,
\end{align}
at which the box first contains the stable equilibrium at the origin.

A second distinguished radius occurs when the box first intersects
the separatrix associated with the saddle energy $-1$. The largest
energy attained in $U_r^{\mathrm{rec}}$ is
\begin{align}
\label{eq:double_pendulum_max_box_energy}
    H_{\max}(r)
    =
    (0.2+r)^2
    +\frac{1}{2}r^2
    +(0.2+r)r
    -3\cos r.
\end{align}
Consequently, the second critical radius is determined by
$H_{\max}(r_2)=-1$, giving
\begin{align}
\label{eq:double_pendulum_r2}
    r_2\approx0.6325.
\end{align}

Figure~\ref{fig:double_pendulum_rectangular_losses} shows the
domain-averaged deflated losses of the three learned candidates as the
box expands. The loss of $I_1$, which is associated with the
Hamiltonian, remains much smaller than those of the other candidates,
but it nevertheless increases with $r$. Since the Hamiltonian is a
global conservation law on every domain considered, this increase is
attributable to numerical effects, including finite sampling,
optimization error, and the increasing difficulty of representing the
law accurately on a larger domain. We therefore regard the loss of
$I_1$ as a numerical loss floor against which the other candidates are
compared.

Relative to this loss floor, the third candidate begins to rapidly deteriorate
before the box approaches the equilibrium at $r_1=0.2$. This does not
contradict the flow-box theorem. The theorem guarantees three
independent local conservation laws only on some sufficiently small
neighborhood of the regular center point; it certainly does not imply that these
laws extend throughout every larger box that avoids an equilibrium.
Such an extension requires the one-dimensional orbit foliation inside
the entire box to admit three globally defined, single-valued, smooth
coordinates transverse to the flow. This global flow-box structure can
fail before an equilibrium enters the domain, for example when the
same trajectory passes through the box in multiple disconnected
segments or intersects a transverse section more than once. The
observed increase of the loss of $I_3$ before $r_1$ therefore indicates
that the three-law local representation is no longer resolved on the
whole box, although the numerical experiment does not sharply locate
the precise obstruction. Indeed, while from an interpretation 
perspective, such an identification of the obstruction would be 
desirable, our numerical experiments did not point to a specific
interpretable condition thereof.

The second candidate remains close to the numerical loss floor through
$r_1$ and up to approximately the separatrix radius
$r_2\approx0.6325$. Beyond $r_2$, its loss separates gradually from
that of $I_1$ and continues to increase as the box expands. Thus, the
experiment exhibits an overall transition from three recovered laws
to two and eventually to one. The first transition occurs before the
box reaches the stable equilibrium, while the second becomes
more clearly discernible once the box intersects the separatrix.
Neither transition is sharply localized, although that of the
$I_3$ appears to be considerably sharper than that of $I_2$.
\begin{figure}
    \centering
    \includegraphics[width=.7\columnwidth]{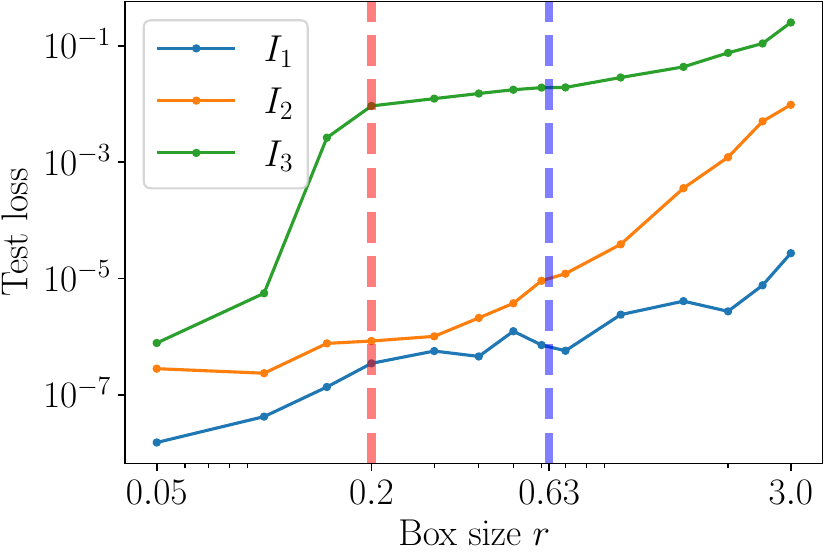}
    \caption{Domain-averaged deflated losses of the three learned candidates for the double pendulum on the expanding boxes $U_r^{\mathrm{rec}}$ centered at $\bx_0=(0,0,0.2,0)$. The vertical dashed lines mark $r_1=0.2$, when the box first contains the stable equilibrium, and $r_2\approx0.6325$, when it first intersects the separatrix of the saddle energy $-1$. The third candidate is lost before $r_1$, while the second remains accurate through $r_2$ and then gradually breaks down.
    }
    \label{fig:double_pendulum_rectangular_losses}
\end{figure}

\subsubsection{Energy sublevel sets growing from the origin}
\label{subsubsec:double_pendulum_energy}

We next grow the training domain directly from the stable equilibrium
at the origin. For an energy level $h>-3$, let
$U_h^{\mathrm{en}}$ denote the connected component containing the
origin of the sublevel set
\begin{align}
\label{eq:double_pendulum_energy_domains}
    U_h^{\mathrm{en}}=\left\{
        (\theta_1,\theta_2,\omega_1,\omega_2):
        H_{\mathrm{dp}}
        (\theta_1,\theta_2,\omega_1,\omega_2)
        \leq h
    \right\}.
\end{align}
Samples are drawn uniformly from $U_h^{\mathrm{en}}$.

Because every energy sublevel set contains the equilibrium at the
origin, the flow-box theorem does not predict three independent
conservation laws in this experiment. Correspondingly, the third
candidate has a large conservation loss from the smallest energy level
considered.

Figure~\ref{fig:double_pendulum_energy_losses} shows the
domain-averaged conservation losses as the energy level $h$ increases.
For each candidate and each value of $h$, the plotted loss is the
smallest value obtained over five independent training runs, which
reduces the influence of unsuccessful optimization runs. The first
candidate, associated with the (existence of) Hamiltonian, retains the smallest loss
throughout the experiment. The second candidate initially has a
somewhat larger but still comparable loss. As the energy sublevel set
expands, however, its loss gradually separates from that of the
first law and continues to increase.

Thus, two conservation laws are numerically resolved on the smallest
energy domains, but only the Hamiltonian remains accurate on larger
ones. The gradual separation of the first and second losses indicates
the breakdown of the second candidate, signaling the loss
of Liouville integrability ---in line with our theoretical
expectations--- but it does not identify a
sharp energy level at which this breakdown occurs.
\begin{figure}
    \centering
    \includegraphics[width=.7\columnwidth]{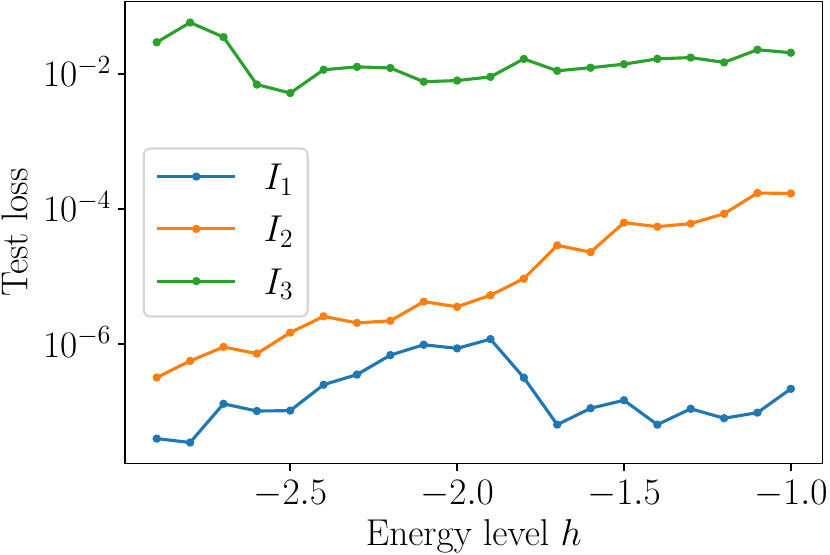}
    \caption{Domain-averaged conservation losses of the learned candidates for the double pendulum on the expanding energy sublevel sets $U_h^{\mathrm{en}}$. The third candidate has a large loss from the outset because every domain contains the equilibrium. The energy candidate remains accurate, while the loss of the second candidate increases gradually as the energy level grows, indicating a transition from two numerically recovered conservation laws to one.} \label{fig:double_pendulum_energy_losses}
\end{figure}

To examine whether the gradual loss of the second candidate can be
localized in phase space, we compare its pointwise conservation loss
with the Smaller Alignment Index (SALI), a standard indicator for
distinguishing regular and chaotic trajectories
\cite{sali}. Given two initially independent deviation
vectors $\bv_1$ and $\bv_2$ evolved along a trajectory and normalized
to $\widehat{\bv}_1$ and $\widehat{\bv}_2$, the index is defined by
\begin{align}
\label{eq:sali}
    \operatorname{SALI}(t)
    =
    \min\left\{
        \left\|\widehat{\bv}_1(t)+\widehat{\bv}_2(t)\right\|,
        \left\|\widehat{\bv}_1(t)-\widehat{\bv}_2(t)\right\|
    \right\}.
\end{align}
For a chaotic trajectory, the two deviation vectors align or
anti-align with the dominant unstable direction, and the SALI rapidly
approaches zero. For a regular trajectory, it typically remains away
from zero. Thus, small SALI values indicate chaotic motion, whereas
larger values indicate regular motion. When $\log_{10}\operatorname{SALI}$
is displayed, more negative values correspond to stronger evidence of
chaos, e.g., blue regions in panels (a) and (d) of Figure~\ref{fig:double_pendulum_sali}.

We evaluate the SALI on the Poincar\'e section
\begin{align}
\label{eq:double_pendulum_poincare_section}
    \Sigma_h
    =
    \left\{
        (\theta_1,\theta_2,\omega_1,\omega_2):
        \theta_2=0,\;
        \omega_2>0,\;
        H_{\mathrm{dp}}=h
    \right\}.
\end{align}
Figure~\ref{fig:double_pendulum_sali} compares the resulting SALI maps
with the pointwise conservation loss of the second learned candidate
for several independent training runs. Although the SALI maps reveal
increasingly complicated regular and chaotic structures as the energy
increases, these structures do not show a consistent spatial
correspondence with the conservation-loss fields. In particular,
regions of relatively large SALI do not systematically coincide with
regions of large conservation loss, and the low-SALI chaotic
structures are not consistently reproduced by the loss fields.
Moreover, the detailed conservation-loss patterns vary between
independent training runs. Accordingly, despite the conceptual
appeal of such a putative connection between the chaos emergence
and the additional conservation loss failure, such a connection
is not brought forth by our data.

The comparison therefore does not identify a sharp dynamical boundary
across which the second conservation law fails. Instead, the
domain-averaged results indicate a gradual deterioration of the second
candidate as the energy sublevel set expands. Neural deflation detects
the eventual reduction from two numerically resolved conservation laws
to the single global energy, but the breakdown cannot  be sharply
localized using SALI (at least not at the resolution of our investigation here), leaving room for further mathematical and computational investigation.


\begin{figure}
    \centering
    \begin{subfigure}[t]{0.3\textwidth}
        \centering
        \includegraphics[width=\textwidth]{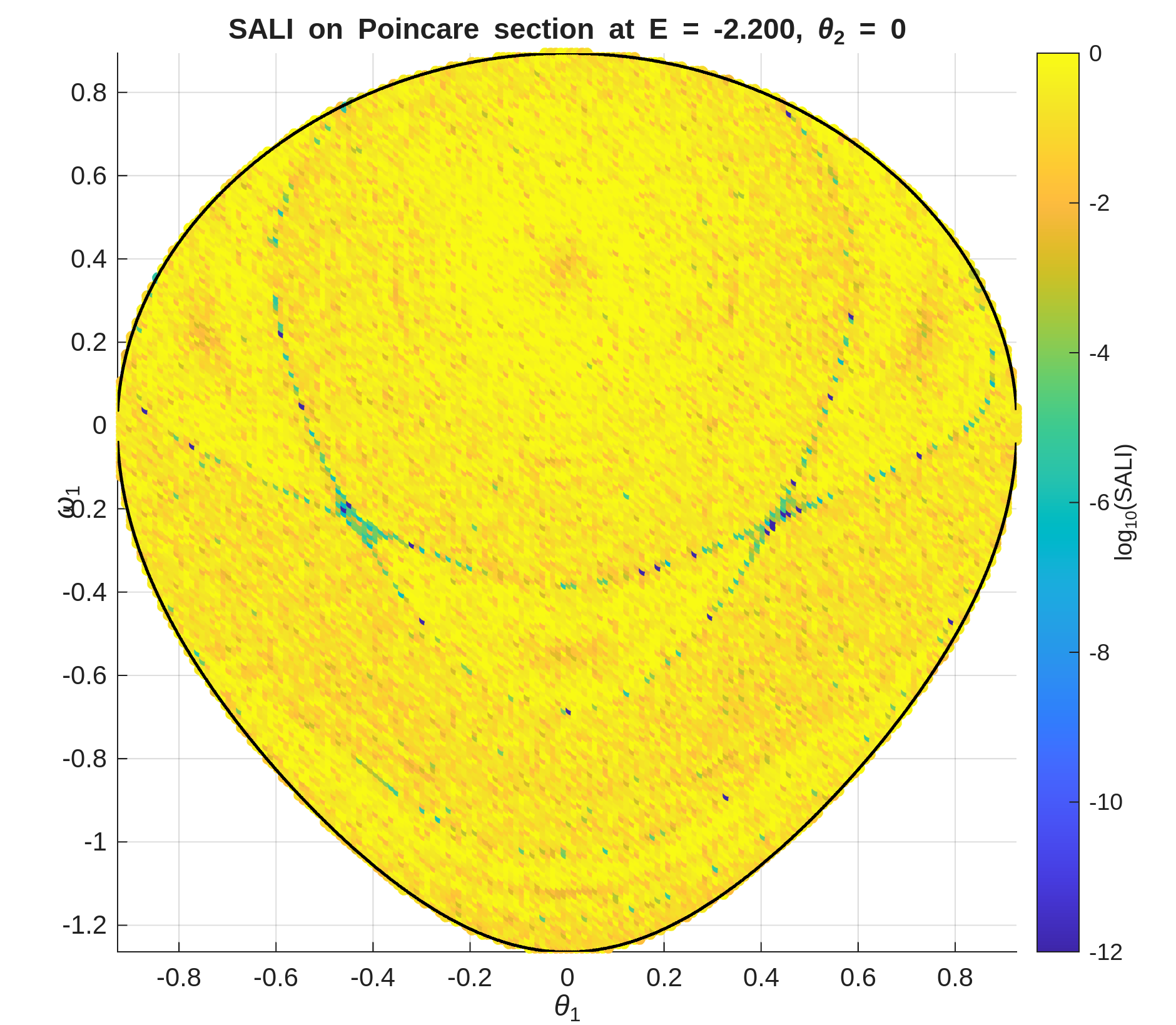}
        \subcaption{SALI, $E=-2.2$}
    \end{subfigure}
    \begin{subfigure}[t]{0.34\textwidth}
        \centering
        \includegraphics[width=\textwidth]{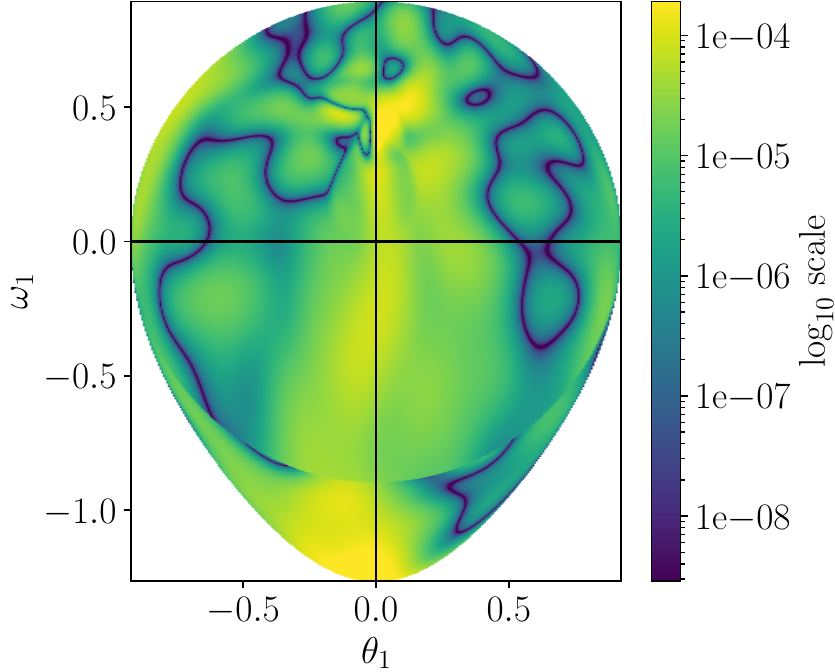} \subcaption{Conservation loss, trial 1.}
    \end{subfigure}
    \begin{subfigure}[t]{0.34\textwidth}
        \centering
        \includegraphics[width=\textwidth]{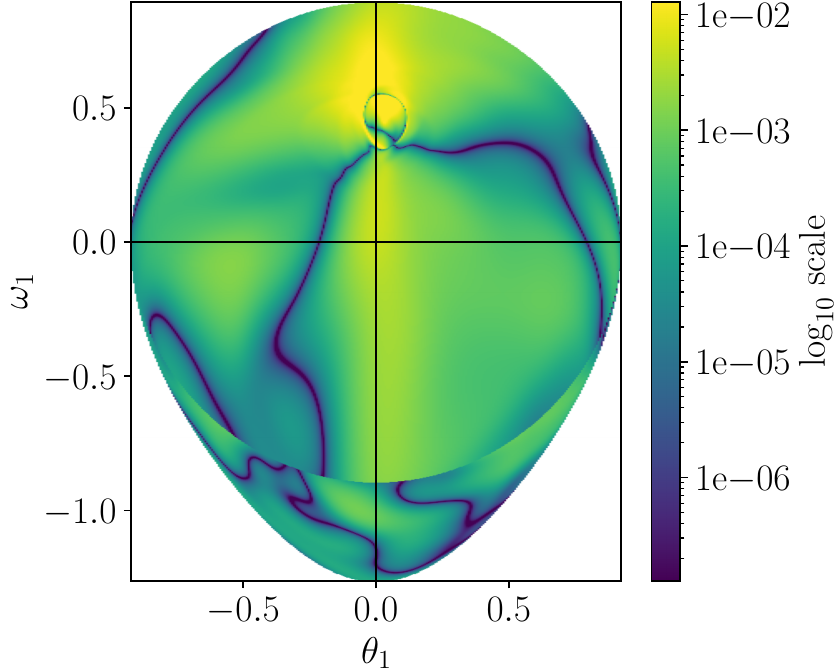} \subcaption{Conservation loss, trial 2.}
    \end{subfigure}    

    \begin{subfigure}[t]{0.3\textwidth}
        \centering
        \includegraphics[width=\textwidth]{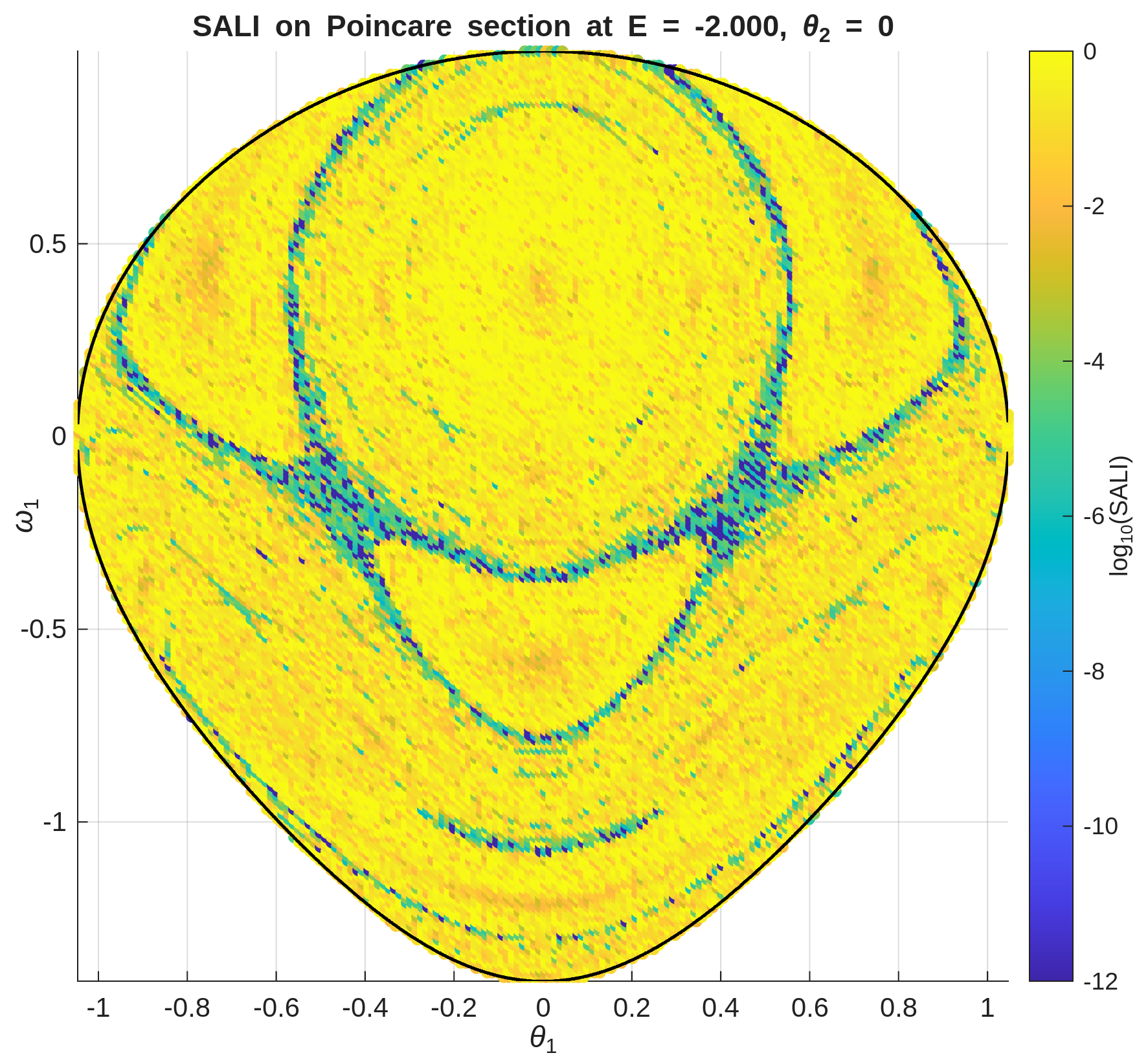}
        \subcaption{SALI, $E=-2.0$}
    \end{subfigure}
    \begin{subfigure}[t]{0.34\textwidth}
        \centering
        \includegraphics[width=\textwidth]{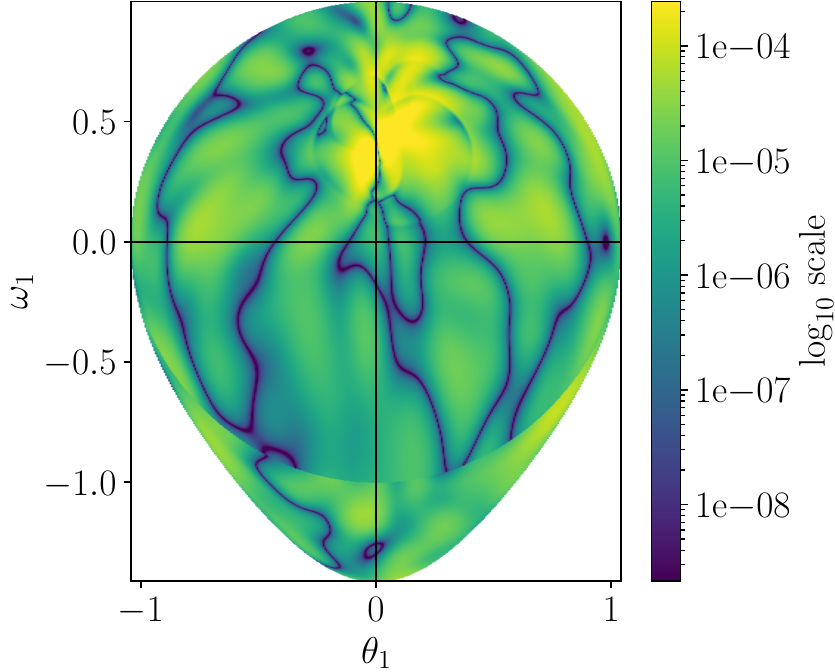} \subcaption{Conservation loss, trial 1.}
    \end{subfigure}
    \begin{subfigure}[t]{0.34\textwidth}
        \centering
        \includegraphics[width=\textwidth]{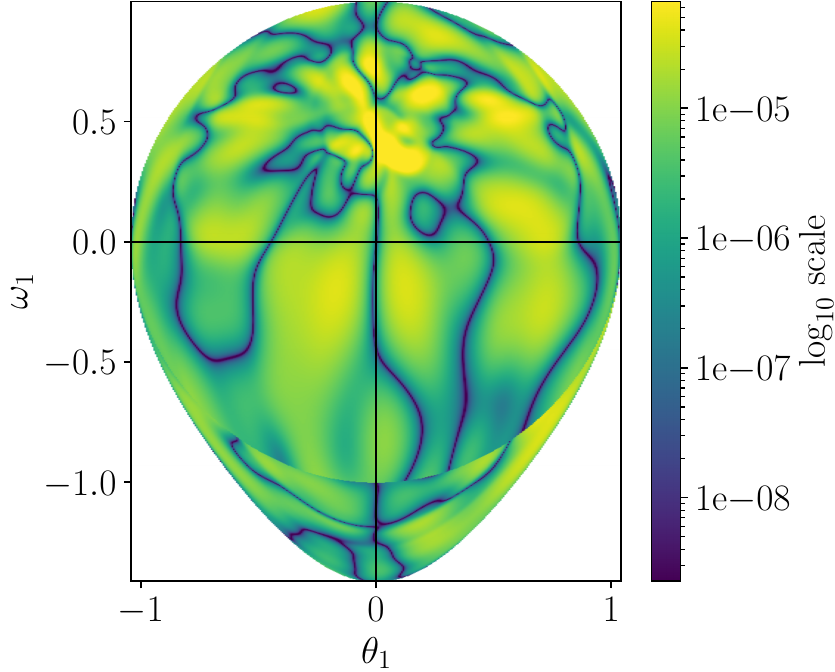} \subcaption{Conservation loss, trial 2.}
    \end{subfigure}       
    \caption{SALI maps and pointwise conservation losses of the second conservation law for the double pendulum at $E=-2.2$ (top row) and $E=-2.0$ (bottom row). In each row, the first panel shows the SALI on the Poincar\'e section, while the second and third show conservation losses from two independent training runs. The SALI structures do not consistently align with regions of large loss, and the loss patterns vary between runs.} \label{fig:double_pendulum_sali}
\end{figure}


\section{\label{sec:conclusion}Conclusion and future work}
We formulated conservation-law discovery as a local-to-global
continuation problem. Although a $d$-dimensional dynamical system
admits $d-1$ independent conservation laws near every regular point,
only a subset of these laws may extend over larger regions of phase
space. To study this transition, we applied neural deflation
over a nested family of expanding domains and used the
resulting validation losses to estimate how the number of independent
conservation laws changes with domain size.

The numerical experiments show that the framework recovers the local
laws predicted by theory and detects their successive breakdown as the
domain expands, interpolating in this way between the
$d-1$ local count of the Frobenius integrability (and the flow-box
theorem) and the asymptotic number for integrable (bearing 
$d/2$ conservation laws in involution)
or non-integrable systems. In the elementary systems considered, pointwise losses sharply
identified when a local law is obstructed by an equilibrium while
distinguishing global laws that extend across the same singular set.
For the H\'enon--Heiles system, the persistence scale of the second
candidate differentiates integrable, near-integrable, and
non-integrable regimes. Nevertheless, the neural deflation technique
considered could not be used as a {\it sharp} detector for 
integrability; rather, it illustrated the desired change of
behavior near the integrable limit. For the double pendulum, the losses revealed a
gradual reduction to the global energy law, but the breakdown was not found to be
sharply localized and it did not correlate directly with the SALI structure
(and the associated chaoticity emergence). These examples demonstrate
that the loss of a conservation law may arise either through a
localized obstruction or through a more diffuse change across phase
space.

The proposed diagnostics remain numerical. Their interpretation
depends on approximation, sampling, and optimization errors, so
candidate losses must be compared with an empirical loss floor (e.g., Hamiltonian) rather
than a universal threshold. Moreover, persistence over the tested
domains does not by itself prove that a learned law extends globally,
and pointwise loss patterns may vary across independent training runs.

Future work will focus on improving the reliability and resolution of the expanding domain,
the continuation procedure through adaptive sampling, and more
systematic uncertainty estimates across training runs. A complementary
theoretical direction is to relate decreases in the number of
conservation laws to geometric properties of the orbit foliation and
to specific dynamical obstructions. 
Additionally, understanding how the findings considered herein
generalize to higher-dimensional systems (i.e., ones with more
degrees of freedom) would be a relevant and worthwhile avenue
for future exploration.
Such developments could turn
expanding-domain conservation-law discovery into a more quantitative
tool for probing global structure in nonlinear dynamical systems.

\begin{acknowledgments}
J.A. acknowledges support from NSF grant PHY-2316622. I.G.K. acknowledges support from the U.S. Department of Energy. P.G.K. acknowledges support from NSF grant PHY-2408988 and the Simons Foundation (SFI-MPS-SFM-00011048). This research was partly conducted while P.G.K. was  visiting the Okinawa Institute of Science and
Technology (OIST) through the Theoretical Sciences Visiting Program (TSVP), the University of
Sydney through the visitor program of the Sydney Mathematical Research Institute (SMRI), and the Department of Mechanical Engineering at Seoul National
University through a Fulbright Fellowship. Their support is gratefully acknowledged. W.Z. acknowledges support from NSF grants DMS-2502900 and DMS-2540370 and AFOSR grant FA9550-25-1-0079.
\end{acknowledgments}

\vspace{1em}
\textbf{J.A.:} Conceptualization, Methodology, Formal analysis, Software, Writing -- review \& editing.

\textbf{I.G.K.:} Conceptualization, Methodology, Formal analysis, Writing -- review \& editing.

\textbf{P.G.K.:} Conceptualization, Methodology, Formal analysis, Writing -- review \& editing.

\textbf{W.Z.:} Conceptualization, Methodology, Formal analysis, Software, Writing -- original draft, Writing -- review \& editing.

\section*{Data Availability Statement}

No external datasets were used in this study. The code used to generate the numerical results is openly available on GitHub at \url{https://github.com/mrzw731/local-deflation}.

\appendix

\section{\label{app:proof}Proof of Theorem~\ref{thm:deflation_consistency}}

\begin{proof}
The integrand in \eqref{eq:population_deflated_loss} is nonnegative.
Therefore, $\mathcal{L}_{K}^{\mu}(\theta_K^*)=0$ if and only if, outside a set of zero $\mu$-measure,
\begin{align}
\label{eq:consistency_conservation}
    \ell_{\mathrm{cons}}
    [I_K(\cdot;\theta_K^*);\bx]
    &=0,
    \\
\label{eq:consistency_independence}
    \ell_{\mathrm{ind},K}(\bx;\theta_K^*)
    &>0.
\end{align}
Equation~\eqref{eq:consistency_conservation} is equivalent to
\[
    \mathbf{f}(\bx)\cdot
    \nabla I_K(\bx;\theta_K^*)=0,
\]
so $I_K(\cdot;\theta_K^*)$ is a conservation law. By the definition
of the independence factor,
\eqref{eq:consistency_independence} holds if and only if
$\nabla I_K(\bx;\theta_K^*)$ has a nonzero component orthogonal to the
span of
\[
    \nabla I_1(\bx;\theta_1^*),\ldots,
    \nabla I_{K-1}(\bx;\theta_{K-1}^*).
\]
This is equivalent to
\[
    \operatorname{rank}
    D(I_1,\ldots,I_K)(\bx)=K
\]
for $\mu$-almost every $\bx\in U$.

Conversely, if $I_K(\cdot;\theta_K^*)$ is a conservation law and the
full collection is functionally independent almost everywhere, then
the numerator in \eqref{eq:population_deflated_loss} vanishes and the
denominator is positive almost everywhere. Hence the integrand
vanishes $\mu$-almost everywhere, and
$\mathcal{L}_{K}^{\mu}(\theta_K^*)=0$.
\end{proof}

\section{Finite-time conservation defect}
\label{app:quad}

For the integrable Hamiltonian \eqref{eq:pred_sep}, consider
\begin{equation}
\label{eq:pred_F}
F=2(A_+^2+A_-^2)=H_\star^2+D^2.
\end{equation}
Its Poisson bracket with $H_\star$ vanishes, and
\begin{equation}
\label{eq:pred_bracket}
G\coloneqq\{F,P\}=12y^2(A_+p_++A_-p_-).
\end{equation}
Let $\Phi_\mu^t$ and $\Phi_\star^t$ denote the flows of
$H_\epsilon$ and $H_\star$, respectively. For a fixed time $T$, define
\begin{align}
\mathcal C_2(z_0;\mu,T)
&=\int_0^T\left(\frac{d}{dt}F(\Phi_\mu^t z_0)\right)^2dt,
\label{eq:pred_C2}\\
g(z_0;T)&=\int_0^T G^2(\Phi_\star^t z_0)\,dt.
\label{eq:pred_g}
\end{align}
The squared residual accumulates without cancellation between positive
and negative drift. Since $\{F,H_\epsilon\}=\mu G$,
\begin{equation}
\label{eq:app_exact}
\mathcal C_2(z_0;\mu,T)
=\mu^2\int_0^T G^2(\Phi_\mu^t z_0)\,dt.
\end{equation}

\begin{theorem}
\label{thm:quad}
Let $\mathcal Z$ be compact and $T>0$. Assume that the unperturbed
trajectories from $\mathcal Z$ exist on $[0,T]$. There are constants
$\mu_0,C>0$, depending on $\mathcal Z$ and $T$, such that the perturbed
trajectories also exist on $[0,T]$ and
\begin{equation}
\label{eq:pred_quad}
\sup_{z_0\in\mathcal Z}
\left|\mathcal C_2(z_0;\mu,T)-\mu^2g(z_0;T)\right|
\leq C|\mu|^3
\end{equation}
for $|\mu|<\mu_0$.
\end{theorem}

\begin{proof}
The unperturbed trajectory segments form a compact set. Choose a compact
ball $\mathcal N$ containing this set in its interior. Continuous
dependence on the initial condition and parameter gives $\mu_0>0$ such
that $z_\mu(t)=\Phi_\mu^t z_0$ remains in $\mathcal N$ for
$z_0\in\mathcal Z$, $0\leq t\leq T$, and $|\mu|<\mu_0$.
Let $L>0$ be a Lipschitz constant for $X_{H_\star}$ on $\mathcal N$,
$M=\sup_{\mathcal N}\|X_P\|_2$, and $L_2$ a Lipschitz constant for $G^2$.
Writing $\zeta(t)=\Phi_\star^t z_0$, the integral equations give
\[
\|z_\mu(t)-\zeta(t)\|_2
\leq L\int_0^t\|z_\mu(s)-\zeta(s)\|_2\,ds+|\mu|Mt.
\]
Gr\"onwall's inequality yields
\begin{equation}
\label{eq:app_gronwall}
\sup_{0\leq t\leq T}\|z_\mu(t)-\zeta(t)\|_2
\leq C_1|\mu|,\qquad
C_1=\frac{M(e^{LT}-1)}{L}.
\end{equation}
Using \eqref{eq:app_exact} and \eqref{eq:pred_g},
\begin{align*}
\left|\mathcal C_2-\mu^2g\right|
&\leq\mu^2L_2\int_0^T\|z_\mu(t)-\zeta(t)\|_2\,dt\\
&\leq TL_2C_1|\mu|^3.
\end{align*}
The constants are uniform over $\mathcal Z$.
\end{proof}

On subsets where $g\geq g_0>0$, the estimate also gives
\begin{equation}
\label{eq:pred_offset}
\log_{10}\mathcal C_2
=2\log_{10}|\mu|+\log_{10}g+O(|\mu|).
\end{equation}
The time $T$ is fixed throughout this estimate. For Figure~\ref{fig:hh_predictor}, the initial conditions lie on
$\{y=0,\ p_y>0\}\cap\{H_\star=E\}$, parameterized by
\begin{equation}
\label{eq:pred_py0}
p_{y,0}=\sqrt{2E-x_0^2-p_{x,0}^2},\qquad
\mathcal D=\{x_0^2+p_{x,0}^2<2E\}.
\end{equation}
Here $E=0.04<1/12$. Since $P=0$ on the section, these initial conditions
also have energy $E$ for every perturbed Hamiltonian. The displayed
calculations use St\"ormer--Verlet with step size $h=0.02$ and $T=100$,
centered differences for $\dot F$, and trapezoidal quadrature for the
integrals. The initial grid is uniform in $(x_0,p_{x,0})$.

\begin{figure}[t]
\centering
\includegraphics[width=\columnwidth]{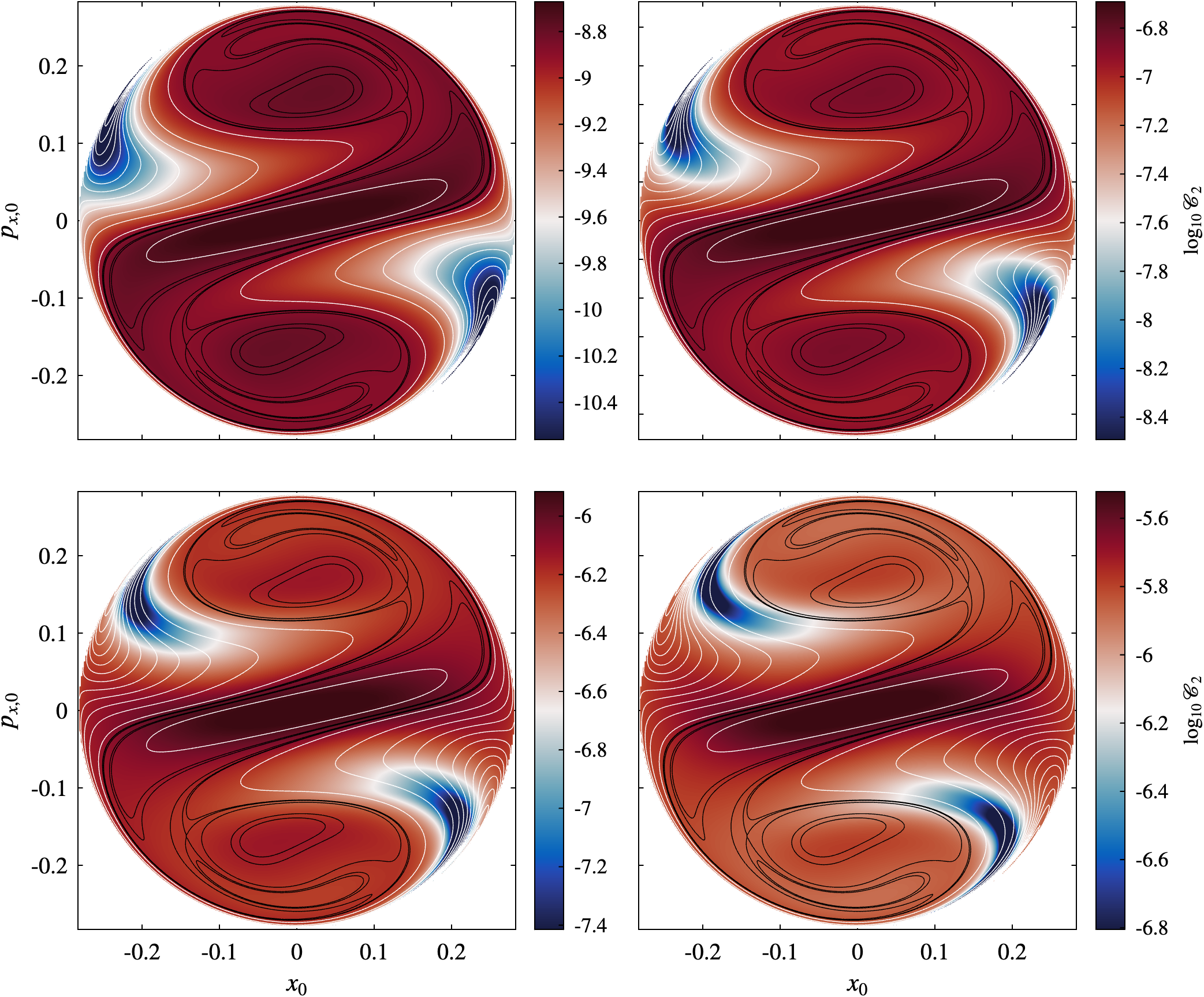}
\caption{\label{fig:hh_predictor}
Finite-time conservation defect of $F$ on the section
$\{y=0,\ p_y>0\}\cap\{H_\star=0.04\}$, with $T=100$.
The panels correspond to $\epsilon=-0.330,-0.300,-0.250,-0.200$
from top left to bottom right. Colors show $\log_{10}\mathcal C_2$
on the perturbed flow, with a separate color scale in each panel.
White curves are contours of $\log_{10}(\mu^2g)$ computed from the
integrable flow. Black curves mark saddle-level contours of $g$.
}
\end{figure}

White curves show contours of $\log_{10}(\mu^2g)$ at levels
$2\log_{10}|\mu|+V_i$, using the same $V_i$ in every panel.
Black curves are level sets of $g$ through its numerically identified
saddle points on $\mathcal D$. The low-defect regions move relative to
the fixed predictor contours as $\epsilon$ increases. These contours
describe the finite-time variation of the chosen observable $F$ and help to visualize where the finite-time drift of the second integral is small or large across the chosen Poincar\'e section.

\bibliography{reference}

@incollection{noether1983invariante,
  title={Invariante variationsprobleme},
  author={Noether, Emmy},
  booktitle={Gesammelte abhandlungen-collected papers},
  pages={231--239},
  year={1983},
  publisher={Springer}
}

@article{SUBRAMANIAN2021105679,
title = {Exploring critical points of energy landscapes: From low-dimensional examples to phase field crystal {PDE}s},
journal = {Communications in Nonlinear Science and Numerical Simulation},
volume = {96},
pages = {105679},
year = {2021},
issn = {1007-5704},
doi = {https://doi.org/10.1016/j.cnsns.2020.105679},
url = {https://www.sciencedirect.com/science/article/pii/S1007570420305098},
author = {P. Subramanian and I.G. Kevrekidis and P.G. Kevrekidis}
}

@book{olver1993applications,
  title={Applications of Lie groups to differential equations},
  author={Olver, Peter J},
  volume={107},
  year={1993},
  publisher={Springer Science \& Business Media}
}

@book{arnold1989classical,
  title={Mathematical methods of classical mechanics},
  author={Arnold, Vladimir Igorevich},
  volume={60},
  year={1989},
  publisher={Springer}
}

@book{goldstein2002classical,
  title={Classical Mechanics},
  author={Goldstein, Herbert and Poole, Charles P. and Safko, John L.},
  edition={3rd},
  year={2002},
  publisher={Addison-Wesley},
  address={San Francisco, CA},
  isbn={978-0201657029}
}

@book{lee2003smooth,
  author    = {Lee, John M.},
  title     = {Introduction to Smooth Manifolds},
  series    = {Graduate Texts in Mathematics},
  volume    = {218},
  publisher = {Springer},
  address   = {New York},
  year      = {2003},
  doi       = {10.1007/978-0-387-21752-9}
}

@article{wetzel2020discovering,
  author  = {Wetzel, Sebastian J. and Melko, Roger G. and
             Scott, Joseph and Panju, Maysum and Ganesh, Vijay},
  title   = {Discovering Symmetry Invariants and Conserved Quantities
             by Interpreting Siamese Neural Networks},
  journal = {Physical Review Research},
  volume  = {2},
  pages   = {033499},
  year    = {2020},
  doi     = {10.1103/PhysRevResearch.2.033499}
}

@article{liu2021machine,
  title={Machine learning conservation laws from trajectories},
  author={Liu, Ziming and Tegmark, Max},
  journal={Physical Review Letters},
  volume={126},
  number={18},
  pages={180604},
  year={2021},
  publisher={APS}
}

@article{ha2021discovering,
  title={Discovering invariants via machine learning},
  author={Ha, Seungwoong and Jeong, Hawoong},
  journal={Physical Review Research},
  volume={3},
  number={4},
  pages={L042035},
  year={2021},
  publisher={APS}
}

@article{liu2022machine,
  title={Machine learning conservation laws from differential equations},
  author={Liu, Ziming and Madhavan, Varun and Tegmark, Max},
  journal={Physical Review E},
  volume={106},
  number={4},
  pages={045307},
  year={2022},
  publisher={APS}
}

@article{lu2023discovering,
  title={Discovering conservation laws using optimal transport and manifold learning},
  author={Lu, Peter Y and Dangovski, Rumen and Solja{\v{c}}i{\'c}, Marin},
  journal={Nature Communications},
  volume={14},
  number={1},
  pages={4744},
  year={2023},
  publisher={Nature Publishing Group UK London}
}

@article{zhu2023machine,
  title={Machine learning of independent conservation laws through neural deflation},
  author={Zhu, Wei and Zhang, Hong-Kun and Kevrekidis, PG},
  journal={Physical Review E},
  volume={108},
  number={2},
  pages={L022301},
  year={2023},
  publisher={APS}
}

@article{chen2025data,
  title={Data-driven discovery of conservation laws from trajectories via neural deflation},
  author={Chen, Shaoxuan and Kevrekidis, Panayotis G and Zhang, Hong-Kun and Zhu, Wei},
  journal={Communications in Nonlinear Science and Numerical Simulation},
  volume={143},
  pages={108563},
  year={2025},
  publisher={Elsevier}
}

@article{burrage1998deflation,
  title={A deflation technique for linear systems of equations},
  author={Burrage, Kevin and Erhel, Jocelyne and Pohl, Bert and Williams, Alan},
  journal={SIAM Journal on Scientific Computing},
  volume={19},
  number={4},
  pages={1245--1260},
  year={1998},
  publisher={SIAM}
}

@article{ojika1987modified,
  title={Modified deflation algorithm for the solution of singular problems. {I}. {A} system of nonlinear algebraic equations},
  author={Ojika, Takeo},
  journal={Journal of mathematical analysis and applications},
  volume={123},
  number={1},
  pages={199--221},
  year={1987},
  publisher={Elsevier}
}

@article{gu2021structure,
  title={Structure probing neural network deflation},
  author={Gu, Yiqi and Wang, Chunmei and Yang, Haizhao},
  journal={Journal of computational physics},
  volume={434},
  pages={110231},
  year={2021},
  publisher={Elsevier}
}

@article{baydin2018automatic,
  title={Automatic differentiation in machine learning: a survey},
  author={Baydin, Atilim Gunes and Pearlmutter, Barak A and Radul, Alexey Andreyevich and Siskind, Jeffrey Mark},
  journal={Journal of machine learning research},
  volume={18},
  number={153},
  pages={1--43},
  year={2018}
}

@article{henon1964applicability,
  title={The applicability of the third integral of motion: some numerical experiments},
  author={H{\'e}non, Michel and Heiles, Carl},
  journal={Astronomical Journal},
  volume={69},
  pages={73},
  year={1964}
}

@article{fordy1991henon,
  title={The {H}{\'e}non-{H}eiles system revisited},
  author={Fordy, Allan P},
  journal={Physica D: Nonlinear Phenomena},
  volume={52},
  number={2-3},
  pages={204--210},
  year={1991},
  publisher={Elsevier}
}

@article{sali,
doi = {10.1088/0305-4470/37/24/006},
url = {https://doi.org/10.1088/0305-4470/37/24/006},
year = {2004},
month = {jun},
publisher = {},
volume = {37},
number = {24},
pages = {6269},
author = {Ch Skokos and Ch Antonopoulos and T C Bountis and M N Vrahatis},
title = {Detecting order and chaos in {H}amiltonian systems by the {SALI} method},
journal = {Journal of Physics A: Mathematical and General}
}

@inproceedings{kaiser2018discovering,
  title={Discovering conservation laws from data for control},
  author={Kaiser, Eurika and Kutz, J Nathan and Brunton, Steven L},
  booktitle={2018 IEEE Conference on Decision and Control (CDC)},
  pages={6415--6421},
  year={2018},
  organization={IEEE}
}

@article{arora2023model,
  title={Model-agnostic machine learning of conservation laws from data},
  author={Arora, Shivam and Bihlo, Alex and Brecht, R{\"u}diger and Holba, Pavel},
  journal={arXiv preprint arXiv:2301.07503},
  year={2023}
}

@article{belyshev2024beyond,
  title   = {Beyond Dynamics: Learning to Discover Conservation Principles},
  author  = {Belyshev, Antonii and Kovrigin, Alexander and
             Ustyuzhanin, Andrey},
  journal = {Machine Learning: Science and Technology},
  volume  = {5},
  number  = {2},
  pages   = {025055},
  year    = {2024},
  doi     = {10.1088/2632-2153/ad4a20}
}

@article{oellerich2024towards,
  title={Towards robust data-driven automated recovery of symbolic conservation laws from limited data},
  author={Oellerich, Tracey and Emelianenko, Maria},
  journal={Machine Learning: Science and Technology},
  volume={5},
  number={3},
  pages={035035},
  year={2024},
  publisher={IOP Publishing}
}

@article{mebratie2025machine,
  title={Machine learning conservation laws of dynamical systems},
  author={Mebratie, Meskerem Abebaw and Nather, R{\"u}diger and von Rudorff, Guido Falk and Seiler, Werner M},
  journal={Physical Review E},
  volume={111},
  number={2},
  pages={025305},
  year={2025},
  publisher={APS}
}

@article{liu2024interpretable,
  title={Interpretable conservation laws as sparse invariants},
  author={Liu, Ziming and Sturm, Patrick Obin and Bharadwaj, Saketh and Silva, Sam J and Tegmark, Max},
  journal={Physical Review E},
  volume={109},
  number={2},
  pages={L023301},
  year={2024},
  publisher={APS}
}

@article{BrownGearhart1971,
  author  = {Brown, K. M. and Gearhart, W. B.},
  title   = {Deflation techniques for the calculation of further solutions of a nonlinear system},
  journal = {Numerische Mathematik},
  volume  = {16},
  pages   = {334--342},
  year    = {1971},
  doi     = {10.1007/BF02165004}
}

@article{elfwing2018sigmoid,
  title={Sigmoid-weighted linear units for neural network function approximation in reinforcement learning},
  author={Elfwing, Stefan and Uchibe, Eiji and Doya, Kenji},
  journal={Neural networks},
  volume={107},
  pages={3--11},
  year={2018},
  publisher={Elsevier}
}

@inproceedings{glorot2010understanding,
  title={Understanding the difficulty of training deep feedforward neural networks},
  author={Glorot, Xavier and Bengio, Yoshua},
  booktitle={Proceedings of the thirteenth international conference on artificial intelligence and statistics},
  pages={249--256},
  year={2010},
  organization={JMLR Workshop and Conference Proceedings}
}

@inproceedings{kingma2015adam,
  author    = {Kingma, Diederik P. and Ba, Jimmy},
  title     = {Adam: A Method for Stochastic Optimization},
  booktitle = {International Conference on Learning Representations (ICLR)},
  year      = {2015},
  url       = {https://arxiv.org/abs/1412.6980}
}

\end{document}